\documentclass[11pt]{article}
\usepackage[margin=1.1in]{geometry}
\usepackage{amsmath,amsthm, amssymb}
\usepackage{graphicx}
\usepackage[colorlinks=true, allcolors=blue]{hyperref}
\usepackage{natbib}
\usepackage{bbm}
\usepackage{algorithm}
\usepackage{tikz}
\usetikzlibrary{arrows,shapes.arrows,shapes.geometric,shapes.multipart,decorations.pathmorphing,positioning,shapes.swigs}
\usepackage{chngcntr}
\usepackage{apptools}
\usepackage{filecontents}
\usepackage{multibib}
\usepackage{titletoc}
\usepackage{caption}
\usepackage{subcaption}
\usepackage{multirow}
\usepackage{listings}
\newcites{supp}{Supplementary References}
\AtAppendix{\counterwithin{lemma}{section}}
\AtAppendix{\counterwithin{theorem}{section}}
\title{%
Nonparametric efficient estimation of the longitudinal front-door functional 
}
\date{\today}

\newcommand\independent{\protect\mathpalette{\protect\independenT}{\perp}}
\def\independenT#1#2{\mathrel{\rlap{$#1#2$}\mkern2mu{#1#2}}}
\newcommand*\diff{\mathop{}\!\mathrm{d}}

\newcommand{\eps}{\varepsilon}

\newtheorem{theorem}{Theorem}
\newtheorem{corollary}{Corollary}
\newtheorem{lemma}{Lemma}
\newtheorem{assumption}{Assumption}

\makeatletter
\newenvironment{breakablealgorithm}
  {
   \begin{center}
     \refstepcounter{algorithm}
     \hrule height.8pt depth0pt \kern2pt
     \renewcommand{\caption}[2][\relax]{
       {\raggedright\textbf{\ALG@name~\thealgorithm} ##2\par}%
       \ifx\relax##1\relax 
         \addcontentsline{loa}{algorithm}{\protect\numberline{\thealgorithm}##2}%
       \else 
         \addcontentsline{loa}{algorithm}{\protect\numberline{\thealgorithm}##1}%
       \fi
       \kern2pt\hrule\kern2pt
     }
  }{
     \kern2pt\hrule\relax
   \end{center}
  }
\makeatother

\usepackage{authblk}
\author[1]{Marie S. Breum\thanks{masb@sund.ku.dk}}
\author[1]{Helene C. W. Rytgaard}
\author[1]{Torben Martinussen}
\author[1,2]{Erin E. Gabriel}
\affil[1]{Section of Biostatistics, University of Copenhagen, Copenhagen, Denmark}
\affil[2]{Pioneer Centre for SMARTbiomed, University of Copenhagen, Copenhagen, Denmark}

\begin{document}

\maketitle

\begin{abstract}
The front-door criterion is an identification strategy for the intervention-specific mean outcome in settings where the standard back-door criterion fails due to unmeasured exposure-outcome confounders, but an intermediate variable exists that completely mediates the effect of exposure on the outcome and is not affected by unmeasured confounding. The front-door criterion has been extended to the longitudinal setting, where exposure and mediator vary over time. However, with the exception of a simple plug-in estimator, no suitable estimation techniques have been proposed. In this work, we derive nonparametric efficient estimators of the longitudinal front-door functional. The estimators accommodate high-dimensional mediators, are multiply robust, and allow for the use of data-adaptive methods for estimating nuisance functions while still providing valid inference. The theoretical properties of the estimators are illustrated in a simulation study, and we apply the estimators to a trial of peanut allergy in infants. 
\end{abstract}

{\it Keywords:} front-door adjustment; longitudinal data; one-step estimation; TMLE; unmeasured confounders

\section{Introduction}\label{sec:intro}

The seminal work of \cite{Pearl1995causal} showed that if a variable exists that (i) completely mediates the effect of exposure on the outcome and (ii) is not affected by unmeasured confounding, then the intervention-specific mean outcome is identified by the front-door formula, even if there are unmeasured exposure-outcome confounders. 
 In \citet{Pearl1995causal} the identification criterion is illustrated through a hypothetical trial investigating the effect of smoking on lung-cancer, where both smoking and lung-cancer are likely confounded by unmeasured genetic predispositions. If we can measure the amount of tar deposited in a subject's lungs, and we believe that smoking does not have any direct effect on lung cancer except that mediated by tar, then the front-door criterion allows for identification of the average causal effect of smoking on lung cancer. Another possible example arises in observational studies of seasonal influenza vaccination, where many unmeasurable behavioral differences may lead to different probabilities of being vaccinated and infected. However, there are also numerous biological variables, many of which will be known and often used in Phase II trials of the vaccine, that are unlikely to be confounded with these behaviors. If these biological variables fully mediate the effect of vaccination on infection rate, the front-door criterion can be used to identify the average causal effect of vaccination on the infection rate

The front-door criterion is a potentially valuable form of identification, specifically in observational settings with unmeasurable and complex confounders of exposure and outcome. Despite this appeal, the front-door criterion has received little attention, mainly due to concerns that the identifying assumptions may be overly restrictive in many applications \citep{cox1995causal, koller2009probabilistic, imbens2020potential}. For instance, in Pearl's smoking and lung cancer example smoking might also affect cancer through mechanisms other than tar deposits, violating (i). In addition, the same unmeasured confounders that confound smoking and lung cancer may also affect tar deposits, thereby violating (ii). Recent work, however, has revisited the front-door criterion and has generated renewed interest in the method. \cite{bhattacharya2022testability} showed that the assumptions are testable under mild assumptions and argue that, when front-door adjustment is feasible, it may offer advantages over instrumental variable approaches, as it allows for nonparametric identification of causal effects without requiring the additional structural restrictions typically imposed by instrumental variable methods. Moreover, \cite{fulcher2020robust} proposed a generalized front-door criterion showing that front-door functional still has a meaningful interpretation as an indirect effect when the assumption that there is no direct effect of treatment on the outcome is replaced by a cross-world independence assumption. \cite{wen2023causaleffectsinterveningvariables} provided an interventionist interpretation of the front-door functional under a dismissible component assumption, also allowing for a direct effect of the exposure on
the outcome not mediated by an intermediate variable. \cite{mao2025front} showed that a wide family of graphs can be reduced to graphs satisfying the ordinary front-door criterion, and argue that the front-door criterion is sometimes preferable to the more complex ID algorithm of \cite{shpitser2006identification}, which may give impractical identifying functionals.

Real-data applications of Pearl's front-door criterion and the generalized front-door criterion have been limited; however, some recent applications in epidemiology demonstrate the potential of the method. \cite{piccininni2023effect} applied the front-door framework to study the causal effect of additional mobile stroke unit (MSU) care on 3-month functional outcomes among stroke patients for whom an MSU was dispatched. The causal effect is not identified using the back-door criterion because decision-making about whether to wait for MSU was not measured. Using time from dispatch to thrombolysis as a mediator, they instead apply the front-door criterion to identify the causal effect. \cite{inoue2022causal} used the front-door framework to investigate the causal effect of chronic pain on all-cause mortality mediated through opioid prescriptions. Moreover, a few applications of the front-door criterion exist in social science and economics where \cite{glynn2017front, glynn2018front} use it to study the effect of signing up for job training on employment outcomes using the indicator of whether the individual actually enrolled in the program as the mediator. \cite{bellemare2024paper} studies the effect of deciding to share an Uber or Lyft ride versus riding solo on tipping, using the indicator of whether the individual actually end up sharing a ride as the mediator.

Often, the exposure and mediator will vary over time. For example, the vaccination scheme described above may require more than one dose. \cite{sjolander2013confounding} generalized the front-door criterion to the longitudinal setting, showing that if there exists a time-varying confounder that is not affected by unmeasured confounding and such that all effects of exposure at time $t$ on the outcome are mediated through its subsequent values, then the intervention specific mean outcome is identified by the longitudinal front-door functional. 
As in the point treatment setting the longitudinal front-door functional still has a meaningful interpretation as an indirect effect when the assumption that there is no direct effect of treatment on the outcome is replaced by a different set of identifying assumptions. 
To our knowledge, applications of the longitudinal front-door criterion remain an open area. This may reflect that, with the exception of a simple plug-in estimator proposed by \cite{sjolander2013confounding}, no suitable estimating techniques exist. By introducing new estimators, this work provides the tools needed to bring such applications within reach.

Even in the point treatment setting, relatively few estimators have been suggested in the literature, compared to the rich literature on back-door estimation. In \cite{fulcher2020robust} the authors proposed a multiply robust one-step estimator of the front-door formula for a dichotomous mediator. \cite{wen2023causaleffectsinterveningvariables} and \cite{guo2023targetedmachinelearningaverage} both proposed targeted minimum loss-based estimators (TMLE) of the front-door formula, which can handle continuous and multivariate mediators. \cite{cavalcante2025nonparametric} proposed various inverse probability weighted (IPW) estimators, including one which is applicable to continuous exposures. 
We build on this work to derive flexible nonparametric estimators of the longitudinal front-door functional, allowing for continuous and multivariate mediators and binary exposures. In particular, we derive the efficient influence function and propose two types of nonparametric efficient estimators: a one-step estimator \citep{vanderVaart2000} and a targeted minimum loss based estimator (TMLE) \citep{vdLRubin2006, VdLRose2011, VdLRose2018}. We study the double robustness properties of the estimators, and we show that the method allows for the use of data-adaptive (machine learning) methods when estimating the nuisance models. Moreover, both of the proposed estimators avoid direct estimation of the conditional mediator density, which is an advantage when, as will typically be the case, the mediator is high-dimensional.

This article is organized as follows. In Section \ref{sec:prelim}, we introduce some notation as well as the identification result.  In Section \ref{sec:rep}, we study alternative representations of the identifying functional based on inverse probability weighting and sequential regression. In Section \ref{sec:estimation}, we propose nonparametric efficient estimators and study their large sample properties. In Section \ref{sec:sim}, we conduct a simulation study to illustrate the theoretical properties of the proposed estimators, and in Section \ref{sec:app} we apply the method to data from a  clinical trial with noncompliance investigating the development of peanut allergy in infants. Some final remarks are provided in Section \ref{sec:dis}. Proofs and technical details are given in the Supplementary Materials.

\section{Preliminaries}\label{sec:prelim}

\subsection{Set-up and notation}
Let $t=0,...,T$ be discrete time points that represent $T$ follow-up visits. We can define the observed data vector as $O=(W, A_0, M_0, ..., A_T,M_T, Y)$ where $W \in \mathbb{R}^d$ is a vector of baseline covariates, $A_t \in \{0,1\}$ is the exposure of interest at time $t$, $M_t \in \mathcal{M}_t \subseteq \mathbb{R}^p$ is a vector of time-varying covariates at time $t$ and $Y \in \mathcal{Y} \subseteq \mathbb{R}$ is the outcome of interest. 
Let $P$ denote the distribution of $O$ that is assumed to lie in a nonparametric statistical model $\mathcal{P}$, and let $O_1,...,O_n$ be a sample of iid observations of $O$. We let $\mathbb{E}$ denote the expectation with respect to $P$ and for a given function $f$, we let $\mathbb{P}_n\{f(O)\}=n^{-1}\sum_{i=1}^nf(O_i)$ denote the empirical average.
We will also use $\bar{X}_t=(X_0,...,X_t)$ to denote the history of an arbitrary variable $X$ and $\underline{X}_t=(X_t, ..., X_T)$ to denote the future of a variable. We use the convention that any variable $X_t$ with $t<0$ is defined as the empty set. 

\subsection{Front-door identification of the intervention-specific mean outcome}
Let $Y(\bar{a}_T)$ be the counterfactual outcome under an intervention that sets $\bar{A}_T=\bar{a}_T \in \{0,1\}^T$. We will assume that the following consistency assumption holds.
\begin{assumption}[Consistency]  \label{ass:consistency}
    If $\bar{A}_t=\bar{a}_t$ then $Y(\bar{a}_t)=Y$ with probability 1 for $t=0,...,T$.
\end{assumption}

In the following, we will focus on identification and estimation of the intervention-specific mean outcome defined as
\begin{align}
\label{eq:target_parameter}
    \Psi^{\bar{a}_T}(P)= \mathbb{E} \left\{Y(\bar{a}_T)\right\},
\end{align}
from which various causal parameters can be defined. A widely used causal parameter is the average causal effect (ACE) of always being exposed versus never being exposed, which is defined as the contrast $\Psi^{\bar{1}_T}(P)-\Psi^{\bar{0}_T}(P)$. Another example of a causal parameter is the population intervention effect (PIE) defined as the contrast between the observed mean outcome and the intervention-specific mean outcome $\mathbb{E}(Y)-\Psi^{\bar{a}_T}(P)$ \citep{hubbard2008population}.

Let $U$ be a set of unmeasured confounders affecting $\bar{A}_T$ and $Y$. The following assumptions on $(O, U)$ will allow for identification of $\Psi^{\bar{a}_T}(P)$. 

\begin{assumption} \label{ass:id}
    Suppose that the following conditional independencies hold for $t=0,...,T$.
\begin{itemize}
    \item[(i)] $M_t \independent U \mid \bar{A}_t, \bar{M}_{t-1}, W$,
    \item[(ii)] $Y \independent \bar{A}_t \mid U, \bar{M}_t, W$,
    \item[(iii)] $Y(\bar{a}) \independent A_t \mid U, \bar{A}_{t-1}, \bar{M}_{t-1}, W$. 
\end{itemize}
\end{assumption}

Together, Assumption \ref{ass:id} (i)-(iii) imply that a) all effects of $A_t$ on $Y$ are mediated through $\underline{M}_t$ and b) $U$ has no direct effect on $M_t$. We say that $\underline{M}_t$ completely blocks the front-door between $A_t$ and $Y$. In Figure \ref{fig:DAGa} we give an example of a DAG where the assumptions above hold for one time-point. This corresponds to the standard front-door DAG.  In Figure \ref{fig:DAGb} we give an example of a DAG where the assumptions above hold for three time-points. In the vaccine study example, this DAG would represent a vaccine requiring three doses, where the effect of vaccination on infection rate is completely blocked by the mediators $M_1, M_2, M_3$ which may represent various biomarkers collected during the trial. 

\begin{figure}[ht]
\centering
\begin{subfigure}{0.4\textwidth}
\centering
        \begin{tikzpicture}
        \node[draw, dashed, circle] (u) at (2, 2) {$\boldsymbol{U}$};
\node (l0) at (-1,1) {$W$};
\node (a0) at (0,0) {$A_0$};
\node (m0) at (2,0) {$M_0$};
\node (y) at (4,0) {$Y$};
\draw[-latex] (a0) -- (m0);
\draw[-latex] (u) -- (a0);
\draw[-latex] (u) -- (y);
\draw[-latex] (m0) -- (y);
              \end{tikzpicture}
               \caption{ }
               \label{fig:DAGa}
\end{subfigure}%
    ~ 
\begin{subfigure}{0.6\textwidth}
\centering
        \begin{tikzpicture}
        \node[draw, dashed, circle] (u) at (2, 2) {$\boldsymbol{U}$};
\node (l0) at (-1,1) {$W$};
\node (a0) at (0,0) {$A_0$};
\node (a1) at (2,0) {$A_1$}; 
\node (m1) at (3,-1) {$M_1$};
\node (m0) at (1,-1) {$M_0$};

\node (a2) at (3.5,0) {$A_2$}; 
\node (m2) at (5,-1) {$M_2$};
\node (y) at (6,-1) {$Y$};
\draw[-latex] (a0) -- (a1);
\draw[-latex] (a0) -- (m0);
\draw[-latex] (a1) -- (m1);
\draw[-latex] (m0) -- (m1);
\draw[-latex] (m1) -- (m2);
\draw[-latex] (m0) -- (a1);
\draw[-latex] (a0) -- (m1);
\draw[-latex] (m2) -- (y);
\draw[-latex] (u) -- (a1);
\draw[-latex] (u) -- (a0);
\draw[-latex] (u) -- (a2);
\draw[-latex] (u) -- (y);

\draw[-latex] (a1) -- (a2);
\draw[-latex] (a1) -- (m2);
\draw[-latex] (m1) -- (a2);
\draw[-latex] (a2) -- (m2);


\draw[->] (a0) to[out=-90, in=-135](m2);
\draw[->] (m0) to[out=-35, in=-160](m2);
\draw[->] (m1) to[out=-45, in=-135](y);

\draw[->] (a0) to[out=35, in=135](a2);
              
\draw[->] (m0) to[out=-45, in=-135](y);
              \end{tikzpicture}
               \caption{ }
               \label{fig:DAGb}
\end{subfigure}
    \caption{(a) Example DAG in a setting where the Assumption 2 holds for one time-point. (b) Example DAG in a setting where the Assumption 2 holds for three time-points. The arrows from $W$ are suppressed, but arrows may exist to all measured variables and between $W$ and  $U$ in either direction.}
    \label{fig:DAG}

\end{figure}

The following assumption ensures that the intervention $\bar{a}_T$ has support in the data. 

\begin{assumption}[Sequential positivity] \label{ass:positivity}
Suppose the following sequential positivity assumptions hold
    \begin{itemize}
        \item[(i)] $\sup  \frac{\mathbbm{1}(\bar{A}_T=\bar{a}_T)}{\prod_{t=0}^T p(a_t \mid W, \bar{M}_{t-1}, \bar{a}_{t-1})} < \infty \quad a.s.$, 
         \item[(ii)] $\sup\prod_{t=0}^T \frac{p(M_t \mid W, \bar{a}_t, \bar{M}_{t-1})}{p(M_t \mid W, \bar{A}_t, \bar{M}_{t-1})} < \infty \quad a.s.$.
    \end{itemize}
\end{assumption}
Assumption \ref{ass:positivity} (i) is the usual sequential positivity assumption for time-varying treatments. It will hold if, at each time $t$,  for every possible history of past covariates there is a positive probability of the exposure level of interest. 
Assumption \ref{ass:positivity} (ii) ensures that there is sufficient overlap between the mediator distribution under the intervention of interest and the observed mediator distribution. 

We then have the following identification theorem.
\begin{theorem}[Identification] \label{theorem:id}
Under Assumptions \ref{ass:consistency}-\ref{ass:positivity}, $\Psi^{\bar{a}_T}(P)$ is identified via the longitudinal front-door functional
\begin{multline}
    \label{eq:F_functional}
     \prod_{t=0}^T \int g_t(m_t \mid w, \bar{a}_t, \bar{m}_{t-1}) \sum_{\bar{a}_T' \in \{0,1\}^T} Q_Y\left(w, \bar{a}_T', \bar{m}_T\right)  \prod_{t=0}^T \pi_t(a_t' \mid w, \bar{m}_{t-1}, \bar{a}_{t-1}') \\ \times p(w) \diff \mu_{M_t}(m_t) \diff \mu_{W}(w),
\end{multline}
where $\pi_t(a_t' \mid w, \bar{a}_{t-1}', \bar{m}_{t-1})$ and $g_t(m_t \mid w, \bar{a}_t, \bar{m}_{t-1})$ are the conditional densities of $A_t$ and $M_t$, $Q_Y( w, \bar{a}_T, \bar{m}_T) \equiv \mathbb{E}\left(Y \mid  w , \bar{a}_T, \bar{m}_T\right)$ is the conditional expectation of $Y$, and $\mu_{\boldsymbol{M_t}}$, $\mu_{W}$ are some dominating measures. 
\end{theorem}
\begin{proof}
    This identification result was first shown in \cite{sjolander2013confounding}. For completeness, we give the proof in Section S1 of the Supplementary Materials. 
\end{proof}
Following \cite{sjolander2013confounding}, we refer to the identifying functional in \eqref{eq:F_functional} as the F-functional. Constructing a plug-in estimator based on this identifying functional requires estimation of possibly high dimensional densities of the mediator as well as calculation of integrals with respect to these densities, which is challenging due to the curse of dimensionality. We will therefore explore other estimation techniques. First, we discuss an alternative interpretation of the F-functional based on a different set of identifying assumptions.

\subsection{Alternative interpretations of the longitudinal front-door functional}

In the point treatment setting \cite{fulcher2020robust} showed that the front-door functional still has a meaningful interpretation as the indirect component of the PIE of \cite{hubbard2008population}, the so-called population intervention indirect effect (PIIE), under a different set of identifying assumptions.
The PIIE captures the extent to which the effect of exposure is mediated by an intermediate variable when exposure is fixed at its natural value. In the longitudinal setting the PIIE would be the contrast $\mathbb{E}(Y)-\Phi^{\bar{a}_T}(P)$ where $\Phi^{\bar{a}_T}(P)=\mathbb{E}\left\{Y(\bar{A}_T, \bar{M}_T(\bar{a}_T))\right\}$ is the potential outcome under an intervention that for $t=0,...,T$ sets $M_t$ to the value it would have taken when exposure history equals $\bar{a}_t$ while keeping $A_t$ at its natural value. Along with the population intervention direct effect (PIDE) the PIIE decomposes the PIE into a direct and an indirect effect:
\begin{align*}
    PIE^{\bar{a}_T} &= \mathbb{E}(Y) -\mathbb{E}\left\{Y(\bar{a}_T)\right\} = \underbrace{\mathbb{E}(Y) - \Phi^{\bar{a}_T}(P)}_{PIIE} + \underbrace{\Phi^{\bar{a}_T}(P) - \mathbb{E}\left\{Y(\bar{a}_T)\right\}}_{PIDE}.
\end{align*}

We show in Section S3 of the Supplementary Materials that $\Phi^{\bar{a}_T}(P)$ is identified by the longitudinal front-door functional when the following exchangeability assumptions hold: (i) $\underbar{M}_t(\bar{a}_T) \independent A_t \mid W, \bar{A}_{t-1}, \bar{M}_{t-1}$ and (ii) $Y(\bar{m}_T) \independent M_t \mid W, \bar{A}_t, \bar{M}_{t-1}$ and (iii) $Y(\bar{m}_T) \independent \bar{M}_T(\bar{a}_T) \mid W$. Importantly, these assumptions do not rule out paths directly from treatment to outcome. For example, the assumptions hold for the NPSEM-IE associated with the DAG in Figure \ref{fig:DAG_PIIE}.

\begin{figure}[h]
\begin{center}
        \begin{tikzpicture}
\node[draw, dashed, circle] (u) at (2, 2) {$\boldsymbol{U}$};
\node (l0) at (-1,1) {$W$};
\node (a0) at (0,0) {$A_0$};
\node (a1) at (2,0) {$A_1$}; 
\node (m1) at (3,-1) {$M_1$};
\node (m0) at (1,-1) {$M_0$};

\node (a2) at (3.5,0) {$A_2$}; 
\node (m2) at (5,-1) {$M_2$};
\node (y) at (6,-1) {$Y$};
\draw[-latex] (a0) -- (a1);
\draw[-latex] (a0) -- (m0);
\draw[-latex] (a1) -- (m1);
\draw[-latex] (m0) -- (m1);
\draw[-latex] (m1) -- (m2);
\draw[-latex] (m0) -- (a1);
\draw[-latex] (a0) -- (m1);
\draw[-latex] (m2) -- (y);
\draw[-latex] (u) -- (a1);
\draw[-latex] (u) -- (a0);
\draw[-latex] (u) -- (a2);
\draw[-latex] (u) -- (y);

\draw[-latex] (a1) -- (a2);
\draw[-latex] (a1) -- (m2);
\draw[-latex] (m1) -- (a2);
\draw[-latex] (a2) -- (m2);
\draw[->] (a0) to[out=-90, in=-135](m2);
\draw[->] (m0) to[out=-35, in=-160](m2);
\draw[->] (m1) to[out=-45, in=-135](y);
\draw[->] (a0) to[out=35, in=135](a2);              
\draw[->] (m0) to[out=-45, in=-135](y);
\draw[-latex, color=red] (a0)  to[out=45, in=110] (y);
\draw[-latex, color=red] (a1)  to[out=40, in=115] (y);
\draw[-latex, color=red] (a2)  to (y);

              \end{tikzpicture}
\end{center}
    \caption{Example DAG with three time-points.}
    \label{fig:DAG_PIIE}
\end{figure}

In Section S3 of the Supplementary Materials we also provide an interventionist interpretation of the longitudinal front-door functional under a dismissible component assumption, extending the work of \cite{wen2023causaleffectsinterveningvariables}. The identification strategy does not rely on the cross-world exchangeability assumption in (iii) and still allows for the presence of a direct effect of the exposure on the outcome not mediated by an intermediate variable.

\section{Representations of the F-functional} \label{sec:rep}
The F-functional  in \eqref{eq:F_functional} admits several nested expectation and  inverse probability weighted (IPW) representations. For reasons that we will expand upon below, we will not pursue estimators based on these representations. However, the representations will be useful for the construction of the estimators that we propose in Section \ref{sec:estimation}.

\subsection{Nested expectation representations}
The F-functional has a number of different nested expectation representations \citep{BangRobins2005}. Here we introduce two representations that will be useful later for the development of nonparametric efficient estimators. First, let $Q^{\bar{a}}_{M_{T+1}}(W, \bar{M}_T)\equiv  \sum_{\bar{a}_T'} Q_Y(W, \bar{a}_T', \bar{M}_T)  \prod_{t=0}^T \pi_t(a_t' \mid  W, \bar{M}_{t-1}, \bar{a}_{t-1}')$ and define recursively for $t=T,T-1,...,0$
\begin{align}
\label{eq:SR1}
       Q^{\bar{a}}_{M_t}(P)(W, \bar{M}_{t-1}) &\equiv \mathbb{E} \left\{Q^{\bar{a}}_{M_{t+1}}(P)(W, \bar{M}_t) \mid W, \bar{A}_t=\bar{a}_t, \bar{M}_{t-1} \right\}. 
\end{align}

Then the F-functional  can be expressed as $\Psi^{\bar{a}_T}(P)= \mathbb{E}\left\{ Q_{M_0}(P)(W)\right\}$, and estimators of $Q^{\bar{a}}_{M_t}(W, \bar{M}_{t-1})$  can be constructed by regressing $Q^{\bar{a}}_{M_{t+1}}(W, \bar{M}_t)$ on the observed data $(W, \bar{A}_t, \bar{M}_{t-1})$ sequentially for $t=T,T-1,...,0$.  More details on the construction of a sequential regression estimator based on \eqref{eq:SR1} are given in Section S2 of the Supplementary Materials. 

For the second nested expectation representation, let $R^{\bar{a}}_{A_{T+1}}(P)(W, \bar{A}_{T}, \bar{M}_{T}) \equiv Q_Y(P)(W, \bar{A}_T, \bar{M}_T)$ and define recursively for $t=T,T-1,...,0$
\begin{align}
\label{eq:SR2_1}
    R^{\bar{a}}_{M_t}(P)(W, \bar{A}_t, \bar{M}_{t-1}) &= \sum_{\bar{a}_t' \in \{0,1\}^t} \mathbbm{1}(\bar{A}_t=\bar{a}_t')  \kappa_t^{\bar{a}}(W,\bar{a}'_t, M_{t-1}), 
\end{align}
where $\kappa_t^{\bar{a}}(W,\bar{a}'_t, M_{t-1}) \equiv \mathbb{E} \left\{ R^{\bar{a}}_{A_{t+1}}(W, \bar{a}_t', \bar{M}_t)  \mid W, \bar{A}_t=\bar{a}_t, \bar{M}_{t-1} \right\}$, and
\begin{align}
    \label{eq:SR2_2}
    R^{\bar{a}}_{A_t}(P)(W, \bar{A}_{t-1}, \bar{M}_{t-1}) &\equiv \sum_{a_t' \in \{0,1\}} R^{\bar{a}}_{M_t}(W, a_t', \bar{A}_{t-1}, \bar{M}_{t-1}) \pi_t(a_t' \mid W, \bar{M}_{t-1}, A_{t-1}).
\end{align}

Then the F-functional can be expressed as $\Psi^{\bar{a}_T}(P)=\mathbb{E} \left\{ R_{A_0}(P)(W) \right\}$. Estimators of $\kappa^{\bar{a}}_t(W, \bar{a}_t', \bar{M}_{t-1})$ can be constructed by regressing $R^{\bar{a}}_{A_{t+1}}(W, \bar{a}_t', \bar{M}_t)$ on the observed data $(W, \bar{A}_t, \bar{M}_{t-1})$ among those with $\bar{A}_t=\bar{a}_t$. This should be done for all $\bar{a}_t' \in \{0,1\}^t$, and then $R^{\bar{a}}_{M_t}(W, \bar{A}_t, \bar{M}_{t-1})$ can be computed using the identity in \eqref{eq:SR2_1}. $R^{\bar{a}}_{A_t}(W, \bar{A}_{t-1}, \bar{M}_{t-1})$ can be computed by integrating out $A_t$ using the estimated propensity score $\pi_t$. More details are given in Section S2 of the Supplementary Materials.

Although it is possible to construct sequential regression estimators based on the nested expectation representations above, these estimators will rely on correct specification of $\pi_t$ and $Q_Y$, as well as the sequential regressions. Correct specification of the sequential regressions is difficult to do using parametric models, especially in the longitudinal setting, and may require data-adaptive (machine learning) algorithms. 
However, if using data-adaptive nuisance model estimates directly in the sequential regression estimator, the estimator will typically inherit the slower than root-$n$ convergence rates. Consequently, standard asymptotic results based on asymptotic normality generally do not apply.

\subsection{IPW representations}
Define the following random variables
\begin{align*}
    H_t^{\bar{a}}(W, \bar{A}_t, \bar{M}_t)  &\equiv \prod_{k=0}^t\frac{ g_k(M_k \mid W, \bar{a}_k, \bar{M}_{k-1})}{g_k(M_k \mid W, \bar{A}_k, \bar{M}_{k-1})},\\
    W_t^{\bar{a}}(W, \bar{A}_t, \bar{M}_{t-1}) &\equiv \prod_{k=0}^t \frac{\mathbbm{1}(A_k=a_k)}{ \pi_k(a_k \mid W, \bar{M}_{k-1}, \bar{A}_{k-1})}.
\end{align*}

The F-functional has (at least) two inverse probability weighted representations given by $\Psi^{\bar{a}_T}(P)=\mathbb{E} \left\{W_T^{\bar{a}}(W, \bar{A}_T, \bar{M}_{T-1})  \sum_{\bar{a}_T'} Q_{M_{T+1}}^{\bar{a}}(W, \bar{M}_T)\right\}$ and $\Psi^{\bar{a}_T}(P)= \mathbb{E} \left\{ H_T^{\bar{a}}(W, \bar{A}_T, \bar{M}_T) Q_Y\left(W, \bar{A}_T, \bar{M}_T\right)\right\}$. In Section S2 of the Supplementary Materials, we show how these representations can be used to construct simple reweighted estimators. However, for a number of reasons, we do not recommend to pursue this estimation strategy. For one, consistency of such estimators relies on correct specification of all nuisance models, and, similar to the sequential regression based estimators, they are not compatible with data-adaptive methods for nuisance estimation. Moreover, inverse probability weighted estimators are sensitive to near-positivity violations. 

In the next section, we will use tools from semiparametric theory to propose estimators which allow for the use of certain data-adaptive methods for nuisance estimation while still achieving asymptotic normality and valid inference.

\section{Estimation}\label{sec:estimation}
In what follows, we propose locally efficient estimators based on the efficient influence function (EIF) of the front-door functional. Unlike estimators based on the nested expectation and IPW representations in the previous section, estimators constructed using the efficient influence function are often multiply robust, i.e. they remain consistent when parts of the nuisance models are inconsistently estimated. In addition, they will often allow for the use of certain data-adaptive methods for nuisance estimation while still achieving asymptotic normality and valid inference.

\subsection{Efficient influence function}
The EIF is a key object in semiparametric efficiency theory as it characterizes the semiparametric efficiency bound of regular asymptotically linear estimators. The following result establishes the efficient influence function of the F-functional. 

\begin{theorem}[EIF] \label{theorem:eif} The efficient influence function for $\Psi^{\bar{a}_T}(P)$ can be represented as:
\begin{align*}
    D^*(P)(O)=& D^*_Y(P)(O) + \sum_{t=0}^T D^*_{M_t}(P)(O) + \sum_{t=0}^T D^*_{A_t}(P)(O) + Q_{M_0}^{\bar{a}}(P)(W) -  \Psi^{\bar{a}_T}(P),
\end{align*}
where
\begin{align}
         D^*_Y(P)(O)\equiv& H_T^{\bar{a}}(W, \bar{A}_T, \bar{M}_T)\left\{ Y - Q_Y(W, \bar{A}_T, \bar{M}_T)\right\}, \label{eq:eif:Y}\\
         D^*_{M_t}(P)(O) \equiv& W_t^{\bar{a}}(W, \bar{A}_t, \bar{M}_{t-1}) \left\{Q^{\bar{a}}_{M_{t+1}}(P)(W, \bar{M}_t) - Q^{\bar{a}}_{M_t}(P)(W, \bar{M}_{t-1})\right\}, \label{eq:eif:M}\\
      D^*_{A_t}(P)(O)\equiv& H_{t-1}^{\bar{a}}(W, \bar{A}_{t-1}, \bar{M}_{t-1}) \left\{R_{M_t}^{\bar{a}}(P)(W,\bar{A}_t,\bar{M}_{t-1}) -R_{A_t}^{\bar{a}}(P)(W, \bar{A}_{t-1},\bar{M}_{t-1}) \right\}. \label{eq:eif:A}
\end{align}

\end{theorem}
\begin{proof}
    See Section S1 of the Supplementary Materials. 
\end{proof}

Before presenting estimators based on the EIF, we discuss estimation of the conditional mediator density ratio $H_T^{\bar{a}}$.

\subsection{Estimation of the mediator density ratio}
When, as is typically the case, the mediator is high-dimensional, estimating the conditional mediator density directly may be challenging. Here we propose a reparameterization of the mediator density which overcomes this problem. 
Using Bayes' Theorem, we have
\begin{align*}
    g_t(M_t \mid \bar{a}_t, \bar{M}_{t-1})= \frac{p(\bar{a}_t, \bar{M}_t)}{p(\bar{a}_t, \bar{M}_{t-1})} =\frac{\prod_{j=0}^t p(a_j \mid W, \bar{a}_{j-1}, \bar{M}_t) p(W, \bar{M}_t)}{\prod_{j=0}^t p(a_j \mid W, \bar{a}_{j-1}, \bar{M}_{t-1}) p(W,\bar{M}_{t-1})}.
\end{align*}
Define $\gamma_{j,t}(A_j \mid W, \bar{A}_{j-1}, \bar{M}_t) \equiv p(A_j \mid W, \bar{A}_{j-1}, \bar{M}_t)$. Then the mediator density ratio can be written
\begin{align*}
    h_t^{\bar{a}}(W, \bar{A}_t, \bar{M}_t)= \prod_{j=0}^t\frac{\gamma_{j,t}(a_j \mid W, \bar{a}_{j-1}, \bar{M}_t)}{\gamma_{j,t-1}(a_j \mid W, \bar{a}_{j-1}, \bar{M}_{t-1})}\frac{\gamma_{j,t-1}(A_j \mid W, \bar{A}_{j-1}, \bar{M}_{t-1})}{\gamma_{j,t}(A_j \mid W, \bar{A}_{j-1}, \bar{M}_t)}.
\end{align*}

Using $\gamma_{t,t-1}(A_t \mid W, \bar{A}_{t-1}, \bar{M}_{t-1})=\pi_t(A_t \mid W,  \bar{A}_{t-1}, \bar{M}_{t-1})$ we can write $H_t^{\bar{a}}(W, \bar{A}_t, \bar{M}_t)= \prod_{k=0}^t h_k^{\bar{a}}(W, \bar{A}_k, \bar{M}_k)$ as
\begin{align}
    \label{eq:dens_ratio}
    H_t^{\bar{a}}(W, \bar{A}_t, \bar{M}_t)= \prod_{k=0}^t\frac{\gamma_{k, t}(a_k \mid W, \bar{a}_{k-1}, \bar{M}_t)}{\pi_k(a_k \mid W, \bar{a}_{k-1}, \bar{M}_{k-1})}\frac{\pi_k(A_k \mid W, \bar{A}_{k-1}, \bar{M}_{k-1})}{\gamma_{k,t}(A_k \mid W, \bar{A}_{k-1}, \bar{M}_t)}.
\end{align}

This means that the conditional mediator density ratio can be estimated using parametric or nonparametric binary outcome regression methods. 
A similar reparameterization of the mediator density was used by \cite{zheng2017longitudinal} in a longitudinal mediation setting. 

Note that $\gamma_{k,t}$ and $\pi_k$ are not variationally independent. 
The models must be compatible in order to be correctly specified.

\subsection{Proposed estimators}
We will use the EIF to construct two estimators: a one-step estimator and a targeted minimum-loss-based estimator (TMLE). The estimators share the same large sample properties but may perform differently in finite samples because the latter is a substitution-based estimator and the former is not.

\subsubsection{One-step estimator}

The one-step estimator is defined as the solution to the estimating equation formed by setting the empirical mean of the estimated efficient influence function to zero
\begin{align}
    \hat{\psi}_n^{OS} \equiv \mathbb{P}_n \left\{ D^*_Y(P_n)(O) + \sum_{t=0}^T D^*_{M_t}(P_n)(O) + \sum_{t=0}^T D^*_{A_t}(P_n)(O) + \hat{Q}_{n, M_0}^{\bar{a}}(W) \right\},
\end{align}
where $P_n = \left\{(\hat{Q}_{Y,n}, \hat{\pi}_{n,t}, \hat{H}_{n,t}^{\bar{a}}, \hat{Q}_{n,M_t}^{\bar{a}}, \hat{R}_{n,M_t}^{\bar{a}}) : t = 0,...,T\right\}$ denotes estimators of the relevant components of the data-generating mechanism. In particular, estimators of $\hat{Q}_{n,M_t}^{\bar{a}}$ and $\hat{R}_{n,M_t}^{\bar{a}}$ can be obtained by performing sequential regressions based on the nested expectation representations in Section \ref{sec:rep}. All estimators are potentially data-adaptive. 

The one-step estimator is simple to compute but has the drawback that estimates can lie outside the parameter space. Moreover, like inverse probability weighted estimators, the one-step estimator is sensitive to near-positivity violations. Below, we describe how to construct a substitution-based estimator which also satisfies the EIF. 

\subsubsection{TMLE}

Targeted minimum loss based estimation (TMLE) \citep{vdLRubin2006, VdLRose2011, VdLRose2018} is a general method for constructing substitution-based estimators of parameters in non- and semiparametric models that solve the efficient influence function. 
Construction of the estimator is based on updating initial estimators $\hat{P}_n$ of the nuisance functions in a targeted manner by minimizing the loss along a least favorable submodel through $\hat{P}_n$. In particular, the loss and least favorable submodel are chosen so that the score of the submodel equals the efficient influence function of the target parameter. Then a substitution based estimator is obtained by evaluating $\Psi^{\bar{a}}$ at the targeted estimator $\hat{P}_n^*$.

To construct our targeting algorithm, we will use the following representation of the efficient influence function
\begin{corollary} \label{cor:DA} The summand in  \eqref{eq:eif:A} can be expressed as
  \begin{multline*}
      D^*_{A_t}(P)(O)=H_{t-1}^{\bar{a}}(W, \bar{A}_{t-1}, \bar{M}_{t-1})Z_t^{\bar{a}}(P)(W, \bar{A}_{t-1}, \bar{M}_{t-1} ) \left\{ A_t- \pi_t(1 \mid W, \bar{M}_{t-1}, \bar{A}_{t-1})\right\},
  \end{multline*}
  where
  \begin{align*}
      Z_t^{\bar{a}}(P)(W, \bar{A}_{t-1}, \bar{M}_{t-1} ) \equiv R_{M_t}^{\bar{a}}(P)(W, 1, \bar{A}_{t-1}, \bar{M}_{t-1})- R_{M_t}^{\bar{a}}(P)(W, 0, \bar{A}_{t-1}, \bar{M}_{t-1}).
  \end{align*}
\end{corollary}

We will assume for now that $Y$ is binary or continuous and bounded, and we will consider the following negative loglikehood loss functions
\begin{align}
    \begin{split}
    \label{eq:loss_Y}
    \mathcal{L}_Y(Q_Y)(O) =& - \left\{ Y \log Q_Y(W, \bar{A}_T, \bar{M}_T) + (1-Y) \log \left(1 -  Q_Y(W, \bar{A}_T, \bar{M}_T) \right)\right\}, 
    \end{split}
    \\
    \begin{split}
    \label{eq:loss_A}
        \mathcal{L}_{A_t}(\pi_t)(O) =& - \big\{ A_t \log \pi_t(A_t \mid W, \bar{M}_{t-1}, \bar{A}_{t-1}) \\
        &\qquad \qquad \qquad \qquad \quad + \left(1-A_t \right)\log \left(1 -  \pi_t(A_t \mid W, \bar{M}_{t-1}, \bar{A}_{t-1}) \right)\big\}, 
    \end{split} \\  
    \begin{split}
    \label{eq:loss_M}
    \mathcal{L}_{M_t}(Q_{M_{t}})(O) =& - \big\{ Q^{\bar{a}}_{M_{t+1}}(W, \bar{M}_t) \log  Q^{\bar{a}}_{M_t}(W, \bar{M}_{t-1}) \\
    & \qquad \qquad \qquad \quad+ \left(1 - Q^{\bar{a}}_{M_{t+1}}(W, \bar{M}_t) \right)\log \left(1 - Q^{\bar{a}}_{M_t}(W, \bar{M}_{t-1})  \right)\big\} ,
    \end{split}
\end{align}

and under these loss functions, we consider the following least favorable submodels
\begin{align}
\begin{split}
    \label{eq:submodel_Y}
    Q_{Y, \varepsilon}(W, \bar{A}_T, \bar{M}_T)=& \text{ expit}\left(\text{logit} (Q_Y(W, \bar{A}_T, \bar{M}_T)) + \varepsilon  H_T^{\bar{a}}(W, \bar{A}_T, \bar{M}_T)\right), 
\end{split}
    \\
\begin{split}
    \label{eq:submodel_A}    
    \pi_{t, \varepsilon}(A_t \mid W, \bar{M}_{t-1}, \bar{A}_{t-1} ) =& \text{ expit}\big(\text{logit}(\pi_t(A_t \mid W, \bar{M}_{t-1}, \bar{A}_{t-1} )) \\
    & \qquad \qquad + \varepsilon H_{t-1}^{\bar{a}}(W, \bar{A}_{t-1}, \bar{M}_{t-1}) Z_t^{\bar{a}}(W. \bar{A}_{t-1}, \bar{M}_{t-1})\big),\end{split}\\
    \label{eq:submodel_M}
    Q_{M_t, \varepsilon}(W,\bar{M}_{t-1})=& \text{ expit}\left(\text{logit} \left(Q^{\bar{a}}_{M_t}(W, \bar{M}_{t-1})\right) + \varepsilon  W_t^{\bar{a}}(W, \bar{A}_t, \bar{M}_{t-1}) \right).
\end{align}
Together, these loss functions and submodels satisfy
\begin{align*}
    \frac{d}{d\varepsilon}\Bigg\vert_{\varepsilon=0}  \mathcal{L}_Y(Q_{Y, \varepsilon})(O) &= D^*_Y(P)(O), \\ 
        \frac{d}{d\varepsilon}\Bigg\vert_{\varepsilon=0}  \mathcal{L}_{A_t}(\pi_{t, \varepsilon})(O) &=  D^*_{A_t}(P)(O), \\ 
    \frac{d}{d\varepsilon}\Bigg\vert_{\varepsilon=0}  \mathcal{L}_{M_t}(Q_{M_t, \varepsilon})(O) &= D^*_{M_t}(P)(O). 
\end{align*}

Let $\hat{H}_{n,t}$, $\hat{\pi}_{n,t}$, $\hat{Q}_{n,Y}$ be initial, potentialy data-adaptive, estimators of $H_t$, $\pi_t$, $Q_Y$. The targeting steps can then be summarized as follows.
\begin{description}
      \item[Update $\mathbf{Q_Y}$:] Evaluate the submodel \eqref{eq:submodel_Y} in the estimators $\hat{Q}_{n,Y}$, $\hat{H}_{n,t}$ and minimize the loss \eqref{eq:loss_Y}. Let $\hat{Q}_{n,Y}^*$ denote the updated estimator. 
      \item[Update $\mathbf{\pi_t}$:] Let $\hat{R}_{n, A_{T+1}}^* =\hat{Q}_{n,Y}^*$. For $t=T,...,0$ in decreasing order
      \begin{itemize}
          \item[(i)] Compute an estimate $\hat{R}_{n,M_t}$ of $R_{M_t}$ based on the identity in \eqref{eq:SR2_1} evaluated in $\hat{R}_{n, A_{t+1}}^*$. The outer expectation can be estimated using data-adaptive methods. 
          \item[(ii)] Evaluate the submodel \eqref{eq:submodel_A}  in the estimators $\hat{R}_{n,M_t}$, $\hat{H}_{n,t}$, $\hat{\pi}_{n,t}$ and minimize the loss \eqref{eq:loss_A}. Let $\hat{\pi}_{n,t}^*$ denote the updated estimator. 
          \item[(iii)] Compute $\hat{R}_{n,A_t}^*$ based on the identity in \eqref{eq:SR2_2} evaluated in $\hat{\pi}_{n,t}^*$, $\hat{R}_{n,M_t}$.  
      \end{itemize}
      \item[Update $\mathbf{Q_{M_t}}$:] Compute $\hat{Q}_{n, M_{T+1}}^*(W, \bar{M}_T)= \sum_{\bar{a}_T'} Q_{n,Y}^*(W, \bar{a}_T', \bar{M}_T) \prod_{t=0}^T \hat{\pi}_{n,t}^*(a_t' \mid W, \bar{M}_{t-1}, \bar{a}_{t-1}')$. For $t=T,...,0$ in decreasing order
      \begin{itemize}
          \item[(i)] Compute an estimate $\hat{Q}_{n, M_t}$ of $Q_{M_t}$ based on \eqref{eq:SR1} evaluated in $\hat{Q}_{n, M_{t+1}}^*$. The outer expectation can be estimated using data-adaptive methods. 
          \item[(ii)]  Evaluate the submodel \eqref{eq:submodel_M}  in the estimators $\hat{Q}_{n,M_{t+1}}^*$, $\hat{Q}_{n,M_t}$, $\hat{\pi}_{n,t}^*$ and minimize the loss \eqref{eq:loss_M}. Let $\hat{Q}_{n,M_t}^*$ denote the updated estimator.  
      \end{itemize}
  \end{description}

We can then estimate the target parameter as $\hat{\psi}_n^{TMLE}= \mathbb{P}_n\left\{\hat{Q}_{n, M_0}^*(W)\right\}$.

More details on the implementation of the TMLE algorithm are given in Section S2 of the Supplementary Materials.

The TMLE solves the efficient influence equation and is asymptotically equivalent to the one-step estimator. However, it potentially has better finite sample properties due to being a substitution estimator, guaranteeing that estimates lie within the bounds of the parameter space.

\subsection{Large sample properties}

Both the one-step estimator and the TMLE solve the efficient influence equation and are asymptotically equivalent. Below, we study the asymptotic distribution of the estimators. 

First, define $\tilde{Q}^{\bar{a}}_{M_t}( \hat{Q}_{n, M_{t+1}})(W, \bar{M}_{t-1})=\mathbb{E}\left\{ \hat{Q}_{n, M_{t+1}}(W, \bar{M}_t) \mid W, \bar{A}_t = \bar{a}_t, \bar{M}_{t-1} \right\}$ and  $\tilde{\kappa}_t^{\bar{a}}(W, \bar{a}_t', \bar{M}_{t-1})=\mathbb{E}\left\{\hat{R}^{\bar{a}}_{n, A_{t+1}}(W, \bar{a}_t', \bar{M}_t) \mid W, \bar{A}_t=\bar{a}_t, \bar{M}_{t-1} \right\}$ for $t=T,...,0$. Let $Q_Y^*$ be the large sample limit of $\hat{Q}_{n,Y}$ and for $t=0,...,T$ let $\pi_t^*, H_t^{\bar{a},*}, Q_{M_t}^{\bar{a},*}, R_{M_t}^{\bar{a},*}$,  be the limits of $\hat{\pi}_{n,t}, \hat{H}_{n,t}^{\bar{a}},  \hat{Q}_{n,M_t}^{\bar{a}}, \hat{R}_{n,M_t}^{\bar{a}}$.
Further, let $\tilde{Q}^{\bar{a},*}_{M_t}=\tilde{Q}^{\bar{a}}_{M_t}(Q_{M_{t+1}}^{\bar{a},*})$ and $\tilde{R}^{\bar{a},*}_{M_t}=\tilde{R}^{\bar{a}}_{M_t}(R_{A_{t+1}}^{\bar{a},*})$.

The following lemma is a direct consequence of Theorem S1 in the Supplementary Material which establishes the second order remainder term. 
\begin{lemma}[Multiple robustness] \label{lemma:DR} $\mathbb{E}\left\{D^*(P^*)(O)\right\} = \Psi^{\bar{a}_T}(P)-\Psi^{\bar{a}_T}(P^*)$ if either of the following conditions holds.
    \begin{itemize}
        \item[(i)] $H_t^{\bar{a},*}=H_t^{\bar{a}}$ and $\pi_t^*=\pi_t$ for $t=0,...,T$,
        \item[(ii)] $Q_Y^*=Q_Y$ and $\pi_t^*=\pi_t$ for $t=0,...,T$,
        \item[(iii)] $Q_{M_t}^{\bar{a},*}=\tilde{Q}^{\bar{a},*}_{M_t}$, $R_{M_t}^{\bar{a},*}=\tilde{R}^{\bar{a},*}_{M_t}$ and $H_t^{\bar{a},*}=H_t^{\bar{a}}$ for $t=0,...,T$.
    \end{itemize}
\end{lemma}

Lemma \ref{lemma:DR} implies that both the one-step estimator and the TMLE are multiply robust estimators.

When estimating the mediator density using the representation in \eqref{eq:dens_ratio} the third condition above becomes redundant as $H_t^{\bar{a},*}=H_t^{\bar{a}}$ implies $\pi_t^*=\pi_t$. Let $\gamma_{j,t}^*$  be the large sample limits of $\hat{\gamma}_{n, j,t}$ for $t =0,...,T$, $j<t$. Then the double robustness properties can be stated as follows. 
\begin{corollary} $\mathbb{E}\left\{D^*(P^*)(O)\right\} = \Psi^{\bar{a}_T}(P)-\Psi^{\bar{a}_T}(P^*)$ if either of the following conditions holds.
    \begin{itemize}
        \item[(i)] $\gamma_{j,t}^* = \gamma_{j,t}$ and $\pi_t^*=\pi_t$ for $t=0,...,T$, $j<t$, 
        \item[(ii)] $Q_Y^*=Q_Y$ and $\pi_t^*=\pi_t$ for $t=0,...,T$,
    \end{itemize}
\end{corollary}

To establish the asymptotic distribution of the estimators consider the following assumptions. 
\begin{assumption} \label{ass:AL}  Suppose the following conditions hold.
\begin{itemize}
    \item[(i)] $D^*(P_n)- \psi^{\bar{a}_T}(P_n)$ belongs to a Donsker class.
    \item[(ii)] Regularity conditions: for $t=0,...,T$
    \begin{itemize}
        \item[(a)] $\prod_{k=0}^t\hat{\pi}_{n,k}(a_k \mid W, \bar{M}_{k-1}, \bar{a}_{k-1}) > \delta >0$.
        \item[(b)] $\hat{R}^{\bar{a}}_{n,M_t}(W, a_t', \bar{A}_{t-1}, \bar{M}_{t-1}) < \eta < \infty$ for $a_t'=0,1$.
    \end{itemize}
    \item[(iii)] All nuisance parameters are estimated at rate faster than $n^{-1/4}$.
\end{itemize}
\end{assumption}

Note that the convergence rates in Assumption \ref{ass:AL} (iii) are much slower than the the parametric rate $n^{-1/2}$, and may be satisfied for certain data-adaptive (machine learning) methods. If using super learners to estimate the nuisance models the rate condition can be satisfied if it is satisfied for one of the candidate estimators in the library \citep{vdL2004Oracle, vdL2007SuperLearner}. The Donsker class condition in (i) puts some restrictions on the complexity of the nuisance function estimators \citep{vanderVaart2000}. This assumption can be relaxed by using sample splitting \citep{chernozhukov2018double, zheng2010asymptotic}.

Asymptotic normality of the estimators is established in the following theorem. 
 \begin{theorem}[Asymptotic linearity] \label{theorem:AL}
     Under Assumption \ref{ass:positivity} and Assumption \ref{ass:AL} (i)-(iii) $\sqrt{n}\left\{\hat{\psi}_n^{OS}-\Psi^{\bar{a}_T}(P)\right\}=\sqrt{n} \mathbb{P}_n D^*(P)(O) + o_p(1)$, i.e., $\hat{\psi}_n^{OS}$ is asymptotically linear with influence function equal to the efficient influence function. Under the same assumptions with $P_n$ replaced by $P_n^*$ $\sqrt{n}\left\{\hat{\psi}_n^{TMLE}-\Psi^{\bar{a}_T}(P)\right\}=\sqrt{n} \mathbb{P}_n D^*(P)(O) + o_p(1)$, i.e., $\hat{\psi}_n^{TMLE}$ is asymptotically linear with influence function equal to the efficient influence function.
 \end{theorem}
 \begin{proof}
     See Section S1 of the Supplementary Materials.
 \end{proof}

Theorem \ref{theorem:AL} implies that Wald-type confidence intervals for the one-step estimator can be constructed as $\hat{\psi}_n^{OS} \pm z_{1-\alpha/2} \sqrt{\hat{\sigma}^2/n}$ where $\hat{\sigma}^2$ is the empirical sample variance of the estimated EIF. Analogously, Wald-type confidence intervals for the TMLE can be constructed as $\hat{\psi}_n^{TMLE} \pm z_{1-\alpha/2} \sqrt{\hat{\sigma}^2/n}$.

\section{Simulation study}\label{sec:sim}
The aim of the simulation studies is to demonstrate the finite-sample performances of the proposed estimators. The R code to replicate the simulation studies is available in our GitHub repository \href{https://github.com/mariesbreum/Longitudinal-front-door-estimation}{mariesbreum/Longitudinal-front-door-estimation}.

\subsection{Set-up}
We simulated data from data-generating distributions with two time-points. In order to illustrate the theoretical properties of the proposed estimators, we considered four different data-generating mechanisms. DGM (1): The mediators are binary and there are no direct effects of exposure on the outcome. DGM (2): The mediators are binary and there is a direct effect of exposure on outcome. DGM (3): The mediators are continuous and there are no direct effects of exposure on the outcome. DGM (4): The mediators are continuous and there is a direct effect of exposure on the outcome. In all data-generating distributions, exposure and outcome are binary and affected by an unmeasured confounder. A more detailed description of the data generating distributions is given in Section S4 of the Supplementary Materials. 

In DGM (1) and DGM (2) the mediator density was estimated directly, and in order to evaluate the robustness of the proposed estimators when some or all of the nuisance models are misspecified, we considered the following scenarios: (a)  All nuisance models are (approximately) correctly specified. (b) $Q_Y$ and the sequential regressions are misspecified. (c) $H_t^a$  and the sequential regressions are misspecified. (d) $\pi_t$ and $Q_Y$ are misspecified. (e) All nuisance models are misspecified. In DGM (3) and DGM (4) the mediator density was estimated based on the parameterization in \eqref{eq:dens_ratio} and we considered the following scenarios: (a*) All nuisance models are (approximately) correctly specified. (b*)  $Q_Y$ and the sequential regressions are misspecified. (c*) $\gamma_{j,t}$  and the sequential regressions are misspecified. (d*) All nuisance models are misspecified.

The (approximately) correctly specified estimators of $Q_Y$, $\pi_t$ and the sequential regressions were obtained using ensemble super learners \cite{vdL2007SuperLearner, VdLRose2011} implemented using the \texttt{SuperLearner} R-package \citep{SuperLearner}.  The library of candidate learners for the super learners included a main effects glm, glm with second order interactions and bayes glm. Each learner was included with and without a correlation based variable screener. Details of the library of candidate learners for the super learners are provided in Section S4 of the Supplementary Materials. Correctly specified estimators of the mediator density $g_t$ were obtained using a correctly specified logistic regression model. Misspecified esimators were obtained by fitting parametric regression models that do not adjust for all relevant variables.  

\subsection{Results}
For each data-generating distribution, we generated $1000$ datasets with a sample size of $n \in \{500, 1000, 2000, 4000\}$. For each simulated data set, we computed the estimators derived in Sections \ref{sec:estimation} as well as the IPW and sequential regression estimators based on the representations described in Section \ref{sec:rep}. We computed the absolute bias scaled by $\sqrt{n}$ of all estimators. The results for DGM (1) are displayed in Figure \ref{fig:sim_bias}. The remaining simulation results for DGMs (2)-(4) can be found in Section S4 of the Supplementary Materials. 

\begin{figure}
    \centering
    \includegraphics[width=0.9\textwidth]{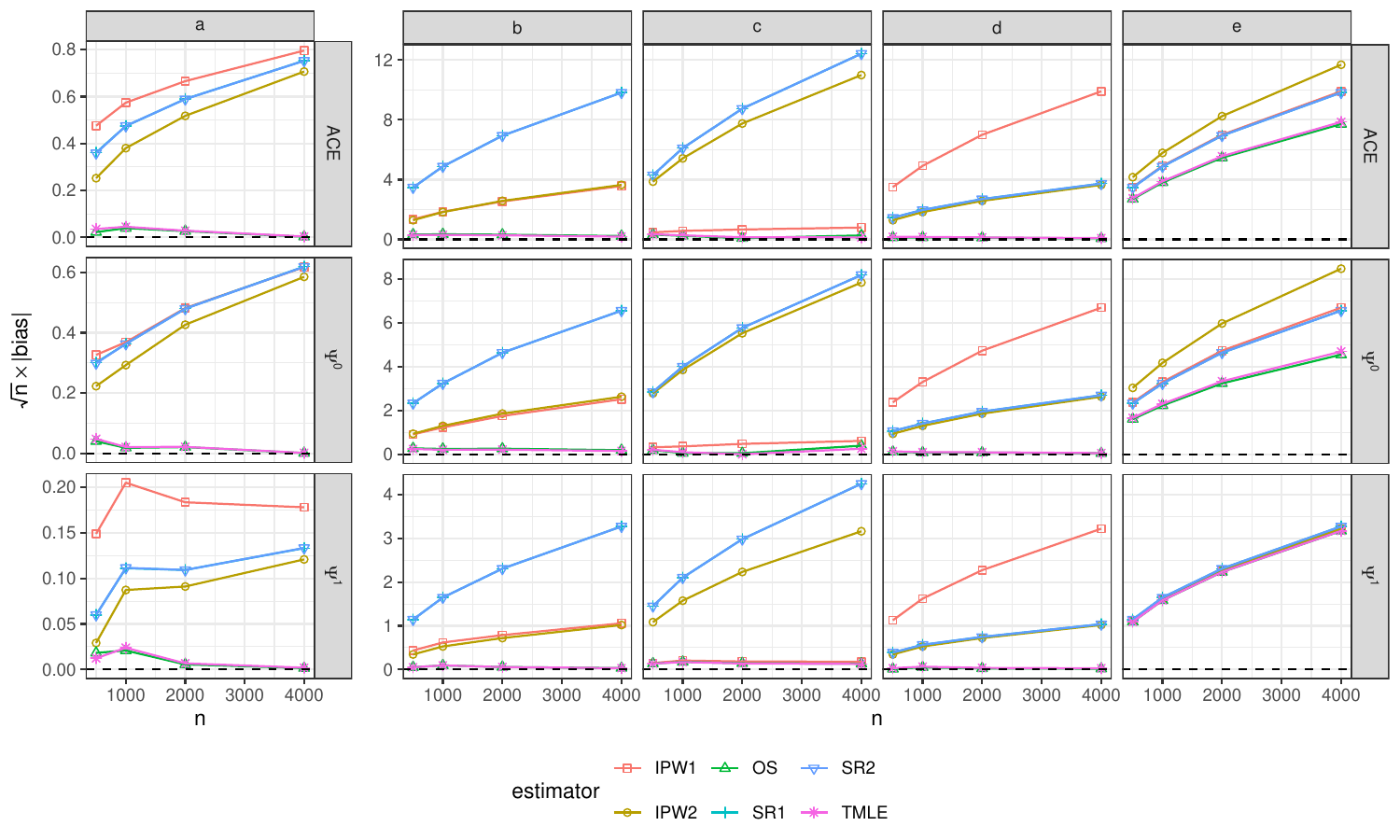}
    \caption{Simulation results for DGM (1). Absolute bias scaled by $\sqrt{n}$.}
    \label{fig:sim_bias}
\end{figure}

The simulation results show that the biases of all estimators are small in (a) when all nuisance models are approximately correctly specified. The biases of the one-step estimator and the TMLE converge to zero at rate $\sqrt{n}$ as expected by Theorem \ref{theorem:AL}, while the biases of the IPW and sequential regression estimators do not.  Further, the TMLE and the one-step estimator remain unbiased in scenarios (b)-(d) as expected by Lemma \ref{lemma:DR}, and are biased in (e) where all nuisance models are misspecified. 

For the one-step estimator and the TMLE we estimated the variance by computing the empirical variance of the efficient influence function and used this to compute the mean square error (MSE) scaled by $n$ and the coverage of a 95\% Wald-type confidence interval. These metrics are displayed 
in Figures \ref{fig:sim_mse} and \ref{fig:sim_cov} below.

\begin{figure}[H]
    \centering
    \includegraphics[width=0.9\textwidth]{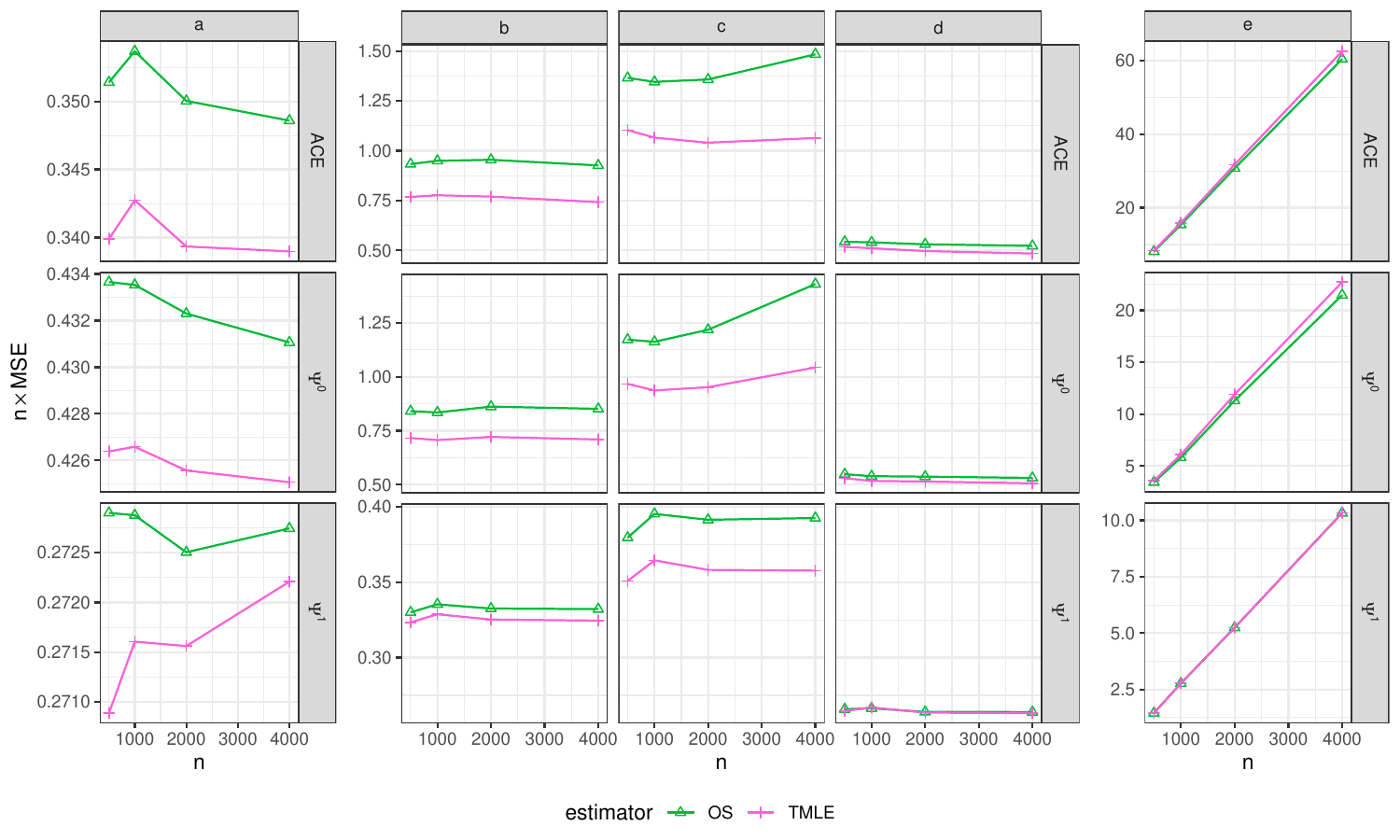}
    \caption{Simulation results for DGM (1). $\mbox{MSE}$ scaled by $n$.}
    \label{fig:sim_mse}
\end{figure}

\begin{figure}[H]
    \centering
    \includegraphics[width=0.9\textwidth]{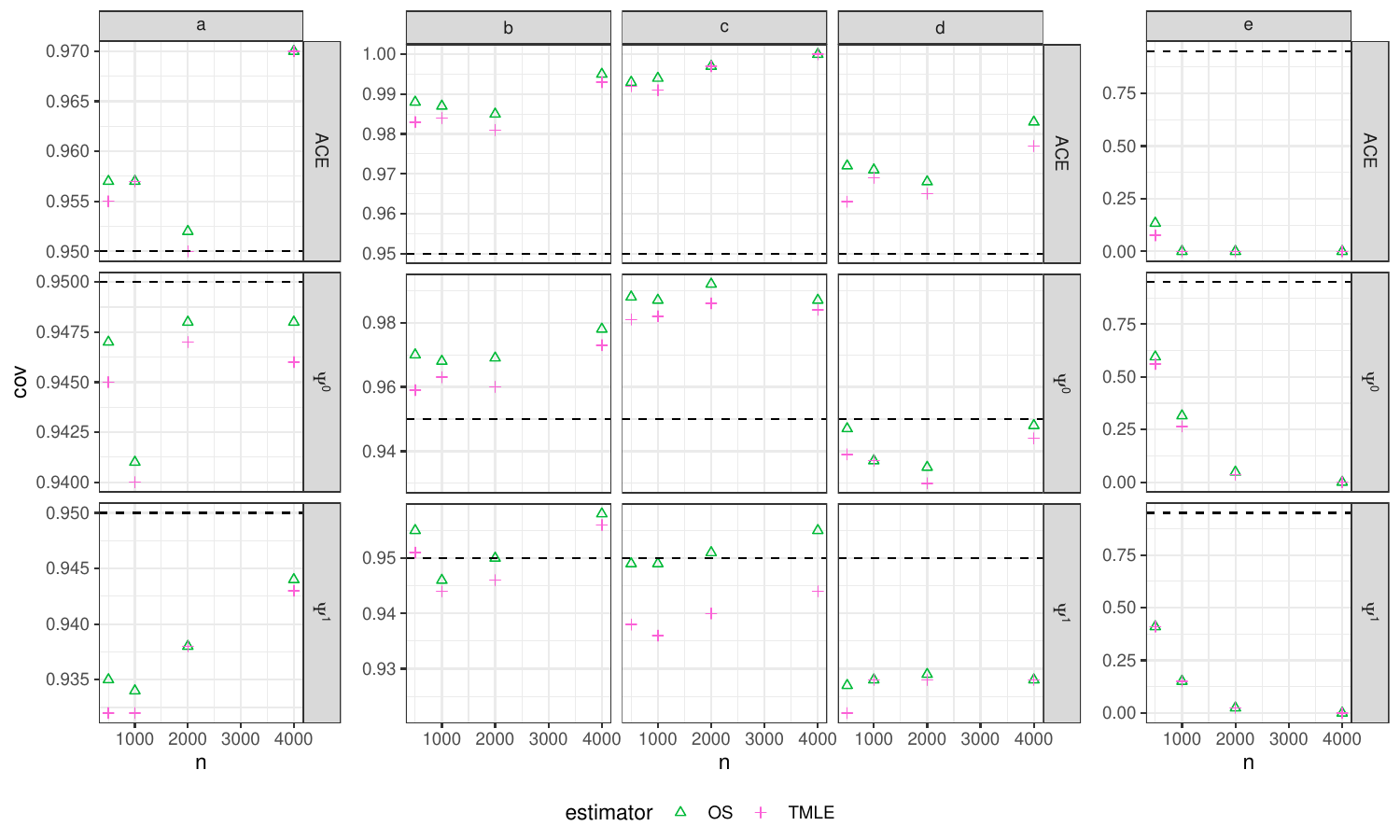}
    \caption{Simulation results for DMG (1). Coverage probability.}
    \label{fig:sim_cov}
\end{figure}

From Figure \ref{fig:sim_mse}, we see that the targeted minimum loss-based estimator yields a better mean squared error than the one-step estimator across (a)-(d). As seen in Figure \ref{fig:sim_cov}, the coverage probabilities for the TMLE and the one-step estimator are nearly identical. In scenario (a), where all models are correctly specified, the coverage probabilities are close to the nominal level. When some nuisance models are misspecified, the influence function based standard errors no longer provide valid inference, according to theory. Surprisingly, the coverage probabilities remain close to the nominal level in (b)-(d). In  (e) the coverage probabilities approach zero as $n$ increases.

\section{Application: Peanut allergy trial}\label{sec:app}

We now apply the longitudinal front-door framework to investigate the causal effect of peanut consumption in infants on development of peanut allergy. Our illustrative data comes from the peanut allergy trial of \cite{du2015randomized}, a two-arm randomized trial which assigned 640 infants between the ages of 4 and 11 months to either consume peanuts or avoid peanuts for the duration of the trial. Let $t=1,2,3$ represent the scheduled visits when participants were 12, 30 and 60 months of age. At baseline and at each visit, serum levels of peanut-specific immunoglobulin antibodies (IgE, IgG, and IgG4) were measured. We let $M_t = (\mbox{IgE}_t, \mbox{IgG}_t, \mbox{IgG4}_t)$ be the time-varying mediator. Adherence was assessed using telephone interviews between visits, which were conducted weekly until visit 1, every 2 weeks between visits 1 and 2, and monthly thereafter. We define exposure $A_t$ to be the indicator of having a weekly average self reported consumption of more than $3$ grams per week between visits $t-1$ and $t$. The primary outcome ($Y$) was peanut allergy at 60 months of age, which was determined by an oral food challenge. Missing mediator values were imputed by carrying forward the last recorded value. Subjects with missing outcome values were removed from the data set. Note that there are very few missing values and this is unlikely to affect the results. We included sex, age, and results of a baseline skin-prick test as baseline covariates, as well as baseline values of IgE, IgG, and IgG4.

The causal effect of self-reported peanut consumption on development of peanut allergy is not identified using the back-door criterion alone because the reason for not adhering to the assigned treatment arm is generally unknown. However, if the measured immunoglobulin antibodies completely mediate the effect of peanut consumption on the primary outcome, and are unconfounded with exposure and outcome conditional on the measured baseline covariates, then the longitudinal front-door criterion allows for identification of the intervention specific mean outcome. 
We estimate the average causal effect (ACE) of high peanut consumption throughout the trial versus low peanut consumption throughout the trial. We also estimate the population intervention effect that contrasts the observed mean outcome and the outcome if everyone had a high peanut consumption throughout the trial ($\mbox{PIE}^1$), and the population intervention effect that contrasts the observed mean outcome and the outcome if everyone had a low peanut consumption throughout the trial ($\mbox{PIE}^0$). 
We note that if peanut consumption has direct effects on the primary outcome not mediated by the measured immunoglobulin antibodies, then the ACE and the PIE are not point identified, but the contrast between the observed outcome mean and the front-door functional may still be interpreted as the indirect component of the PIE, the population intervention indirect effect (PIIE), under a different set of assumptions as discussed in Section \ref{sec:prelim}.
In this application, it seems likely that the effect of peanut consumption on development of peanut allergy is completely mediated by the immune modulation. The main concerns regarding identification are that the measured antibodies may not perfectly capture the true immune modulation process or that the mediator and outcome are affected by an unmeasured confounder such as household peanut exposure levels in which case the unconfounded mediator assumption will be violated. Developing methods for assessing the sensitivity of our results to violations of the unconfounded mediator assumption is an important topic for future research. 

We apply the one-step estimator and the TMLE proposed in Section \ref{sec:estimation}. All nuisance models were estimated using ensemble super learners \citep{vdL2007SuperLearner, VdLRose2011} implemented using the \texttt{SuperLearner} R-package \citep{SuperLearner}.   The library of candidate learners included a main effects glm, glm with second order interactions and bayes glm. Each learner was included with and without a correlation based variable screener.

\begin{table}[]
\centering
\begin{tabular}{cllrrr}
\hline
estimand  &    parameter                 & estimator & \multicolumn{1}{c}{est} & CI low  & CI up    \\ \hline
\multirow{2}{*}{ACE}     & \multirow{2}{*}{$\Psi^{\bar{1}}(P)-\Psi^{\bar{0}}(P)$}     & TMLE      & $-0.0776$              & $-0.102$  & $-0.0527$  \\
                      &   & OS       & $-0.0808$             & $-0.106$  & $-0.0553$  \\
\multirow{2}{*}{$\mbox{PIE}^1$ or $\mbox{PIIE}^1$}     & \multirow{2}{*}{$\mathbb{E}(Y) - \Psi^{\bar{1}}(P)$} & TMLE      & $0.0701$                  & $0.0472$  & $0.0930$    \\
                      &   & OS       & $0.0694$               & $0.0461$  & $0.0927$   \\
\multirow{2}{*}{$\mbox{PIE}^0$ or $\mbox{PIIE}^0$}     &    \multirow{2}{*}{$\mathbb{E}(Y) - \Psi^{\bar{0}}(P)$} & TMLE      & $-0.00752$                & $-0.0131$ & $-0.00198$ \\
                      &   & OS       & $-0.0114$                 & $-0.0173$ & $-0.00553$ \\
\hline
\end{tabular}
\caption{Peanut data trial. est is the estimate; CI lower and CI upper are the upper and lower bounds respectively of a Wald-type 95\% confidence interval.}
\label{tab:peanut_res}
\end{table}

The results are summarized in Table \ref{tab:peanut_res}. We see that the TMLE and the one-step estimator resulted in almost identical point estimates and inferences. The estimated ACE is negative, indicating that peanut consumption has a protective effect on the development of peanut allergy compared to peanut avoidance. 
The estimated $\mbox{PIE}^1$ is non-negative indicating that peanut consumption has a protective effect on the development of peanut allergy compared to the observed risk of peanut allergy. 
Note that since this is a trial, the observed risk of peanut allergy does not correspond to the risk of peanut allergy in a population without intervention. 
We also see that the estimated $\mbox{PIE}^0$ is negative indicating that peanut avoidance increases the risk of peanut allergy compared to the observed risk of peanut allergy (which again is different from the population risk of peanut allergy).

\section{Concluding remarks}\label{sec:dis}

In this work, we developed estimators of the longitudinal front-door functional, which identifies either the intervention specific mean outcome or the counterfactual component of the PIIE, depending on the identification assumptions one is willing to make. Although the longitudinal front-door criterion has been available in the literature for several years, this is the first methodology developed for nonparametric efficient estimation. In particular, we derived the efficient influence function of the longitudinal front-door functional and presented two types of estimators: a one-step estimator and a TMLE. We showed that the estimators remain consistent and asymptotically normal when nuisance models are estimated data-adaptively as long as they satisfy certain rate conditions. This is an important property in longitudinal settings, where simultaneously correctly specifying all nuisance models using parametric models is typically not feasible. We also proposed a reparameterization of the mediator density ratio based on Bayes' theorem, which allowed us to avoid estimating the mediator density directly. This makes our method applicable in settings with continuous and/or high dimensional mediators.

The proposed methodology has some limitations. First, the proposed method is limited to settings with binary exposure at each time point. Extending the methodology to settings with continuous or multi-valued exposure at each time point is a topic for future research. Another important topic for future research is to develop methods for assessing the sensitivity of our results to violations of key identification assumptions such as the unconfounded mediator assumption, and finding Verma or other constraints that can be used to confirm them, under faithfulness, or falsify them.

\section{Disclosure statement}

The authors have no conflicts of interest.

\section{Data Availability Statement}

The peanut consumption trial data used in Section \ref{sec:app} are publicly available and can be downloaded from the Immune Tolerance Network TrialShare website (https://www.itntrialshare.org/, study identifier: ITN032AD).

\section*{Acknowledgments}
MSB and EEG were partially supported by a grant from Novo Nordisk fonden (NNF22OC0076595) and EEG by the Pioneer Centre for SMARTbiomed.

\bibliographystyle{abbrvnat}
\bibliography{bibliography}

\newpage
\centerline{\bf\LARGE Supplementary Materials}
\bigskip
\setcounter{section}{0}
\setcounter{figure}{0}
\setcounter{equation}{0}
\makeatletter
\renewcommand \thesection{S\@arabic\c@section}
\renewcommand\thetable{S\@arabic\c@table}
\renewcommand \thefigure{S\@arabic\c@figure}
\renewcommand \thetheorem{S\@arabic\c@theorem}
\renewcommand \thelemma{S\@arabic\c@lemma}
\renewcommand \thecorollary{S\@arabic\c@corollary}
\renewcommand \theequation{S\@arabic\c@equation}
\makeatother
\setcounter{page}{1}

\section{Proofs}

\subsection{Identification} \label{app:Identification}

We have
\begin{align}
\label{eq:ID1}
    \Psi^{\bar{a}_T}(P) &= \iint \mathbb{E}\left\{ Y(\bar{a}_T) \mid u, w \right\} p(u, w) \diff u \diff w \nonumber \\
    &= \iint \mathbb{E}\left\{ Y(\bar{a}_T) \mid u, w, a_0 \right\} p(u, w) \diff u \diff w \nonumber \\
    &=\iiint \mathbb{E}\left\{ Y(\bar{a}_T) \mid u, w, a_0, m_0 \right\} p(m_0 \mid u, w, a_0) p(u, w) \diff m_0 \diff u \diff w \nonumber \\
    &= \dots \nonumber \\
    &= \iiint  \mathbb{E}\left\{ Y(\bar{a}_T) \mid u, w, \bar{a}_T,\bar{m}_T \right\} \prod_{t=0}^T p(m_t \mid u, w, \bar{a}_T, \bar{m}_{t-1}) p(u, w) \diff \bar{m}_T \diff u \diff w \nonumber \\
    &=\iiint  \mathbb{E}\left( Y \mid u, w, \bar{a}_T,  \bar{m}_T \right) \prod_{t=0}^T p(m_t \mid u, w, \bar{a}_T, \bar{m}_{t-1}) p(u, w)  \diff \bar{m}_T \diff u \diff w \nonumber  \nonumber \\
    &=\iint \underbrace{\int \mathbb{E}\left( Y \mid u, w, \bar{a}_T, \bar{m}_T \right) p(u \mid w) \diff u}_{\equiv (*)} \prod_{t=0}^T g_t(m_t \mid w, \bar{a}_T, \bar{m}_{t-1}) p(w) \diff \bar{m}_T \diff w,
\end{align}
where the first equality follows from the law of total probability (LTP), the second from Assumption 2 (iii), and the third from the LTP. We repeat the arguments in steps two and three until we arrive at the expression on the right-hand side of the fifth equality.  The sixth equality follows from Assumption 1 and the seventh from Assumption 2 (i) and by rearranging. 

We can write $p(u \mid w)$ as
\begin{align}
\label{eq:ID2}
    p(u \mid w) &= \int p(u \mid w, a_0') \pi_0(a_0' \mid w)  \diff a_0' \nonumber \\
    &=  \int  p(u \mid w, a_0', m_0) \pi_0(a_0' \mid w) \diff a_0' \nonumber \\
    &= \dots \nonumber \\
    &= \int p(u \mid w, \bar{a}_T', \bar{m}_T) \prod_{t=0}^T \pi_t(a_t' \mid w, \bar{m}_{t-1}, \bar{a}_{t-1}')\diff \bar{a}_T',
\end{align}
where the first equality follows from the LTP, the second equality from Assumption 2 (i), and the last equality from repeating the arguments in steps 1 and 2.  

We can now write $(*)$ as
\begin{align}
\label{eq:ID3}
    (*) &= \iint \mathbb{E}\left( Y \mid u, w, \bar{m}_T, \bar{a}_T \right) p(u \mid w, \bar{a}_T', \bar{m}_T) \prod_{t=0}^T \pi_t(a_t' \mid w,  \bar{m}_{t-1}, \bar{a}_{t-1}')  \diff \bar{a}_T' \diff u \nonumber \\
    &=\iint  \mathbb{E}\left( Y \mid u, w, \bar{m}_T, \bar{a}_T' \right) p(u \mid w, \bar{a}_T', \bar{m}_T) \diff u \prod_{t=0}^T \pi_t(a_t' \mid w,  \bar{m}_{t-1}, \bar{a}_{t-1}') \diff \bar{a}_T' \nonumber \\
    &=\int Q_Y\left( Y \mid w, \bar{a}_T', \bar{m}_T \right)  \prod_{t=0}^T \pi_t(a_t' \mid w,  \bar{m}_{t-1}, \bar{a}_{t-1}') \diff \bar{a}_T', 
\end{align}
where the first equality follows from inserting \eqref{eq:ID2}, the second equality from Assumption 2 (ii), and the third by integrating out $u$. 

Combing \eqref{eq:ID1} and \eqref{eq:ID3} gives the identifying functional in Theorem 1. 

\subsection{EIF}  \label{app:EIF}
There are different strategies for deriving efficient influence functions in the literature. Here we will compute the Gâteaux derivative in the direction of a point mass contamination as described by \cite{Hines2022demystifying, kennedy2023semiparametric}. The Gâteaux derivative is the pathwise derivative
\begin{align*}
    \frac{\partial}{\partial \varepsilon} \Big\vert_{\varepsilon=0} \Psi(P_\varepsilon) ,
\end{align*}
where $P_\varepsilon =  (1-\varepsilon) P + \varepsilon \delta_{o'}$ and $\delta_{o}$ is the Dirac measure at $O=o$.

For simplicity, we consider $T=1$ and assume that $O=(W, A_0, M_0, A_1, M_1, Y)$ are discrete. Then the target parameter can be written
\begin{multline*}
    \Psi^{\bar{a}_T}(P)=\sum_{w, \bar{m}_1} p(m_1 \mid w, \bar{a}_1, m_0)p(m_0 \mid w, a_0) \sum_{\bar{a}_1'} \sum_y y  p\left(y \mid w, \bar{a}_1', \bar{m}_1\right) \\  \times p(a_1' \mid w, a_0', m_0)p(a_0'\mid w) p(w).
\end{multline*}

In the discrete model $\delta_{o} = \mathbbm{1}(O=o)$ and the Gâteaux derivative can be computed as follows
\small{
\begin{align*}
   &\frac{\partial}{\partial \varepsilon} \Big\vert_{\varepsilon=0} \Psi^{\bar{a}_T}(P_\varepsilon) \\
   =&\frac{\partial}{\partial \varepsilon} \Big\vert_{\varepsilon=0} \Bigg\{\sum_{w, \bar{m}_1} p_{\varepsilon}(m_1 \mid w, \bar{a}_1, m_0)p_{\varepsilon}(m_0 \mid w, a_0) \\
   & \qquad \qquad \quad \times\sum_{\bar{a}_1'} \sum_y y p_{\varepsilon}\left(y \mid w, \bar{a}_1', \bar{m}_1\right) p_{\varepsilon}(a_1' \mid w, a_0', m_0)p_{\varepsilon}(a_0'\mid w) p_{\varepsilon}(w) \Bigg\}\\
   =&\sum_{w, \bar{m}_1} \frac{\mathbbm{1}(W=w, \bar{A}_1=\bar{a}_1, M_0=m_0)}{p(a_1 \mid w, m_0, a_0)p(a_0\mid w)} \left\{\mathbbm{1}(M_1=m_1)-p(m_1 \mid w, \bar{a}_1, m_0)\right\} \\
   &\quad \times \sum_{\bar{a}_1'} Q_Y\left(w, \bar{a}_1', \bar{m}_1\right) p(a_1' \mid w, a_0', m_0)p(a_0' \mid w)\\
   &+\sum_{w, \bar{m}_1}  p(m_1 \mid w, \bar{a}_1, m_0) \frac{\mathbbm{1}(W=w, A_0=a_0)}{p(a_0\mid w)} \left\{\mathbbm{1}(M_0=m_0)-p(m_0 \mid w, a_0)\right\} \\
   &\qquad \times \sum_{\bar{a}_1'} Q_Y\left(w, \bar{a}_1', \bar{m}_1\right) p(a_1' \mid w, a_0', m_0)p(a_0'\mid w)  \\
   &+ \sum_{w, \bar{m}_1} p(m_1 \mid w, \bar{a}_1, m_0)p(m_0 \mid w, a_0) \sum_{\bar{a}_1'} \frac{\mathbbm{1}(W=w, \bar{A}_1=\bar{a}_1', \bar{M}_1=\bar{m}_1)}{p(m_1 \mid w, \bar{a}_1', m_0)p(m_0 \mid w, a_0')} \left\{Y-Q_Y\left(w, \bar{a}_1', \bar{m}_1\right) \right\} \\
   &+ \sum_{w, \bar{m}_1} p(m_1 \mid w, \bar{a}_1, m_0)p(m_0 \mid w, a_0) \sum_{\bar{a}_1'} Q_Y\left(w, \bar{a}_1', \bar{m}_1\right) \frac{\mathbbm{1}(W=w,A_0=a_0', M_0=m_0)}{p(m_0 \mid w, a_0')} \\
   & \qquad \times \big\{\mathbbm{1}(A_1=a_1')-p(a_1' \mid w,a_0', m_0) \big\}  \\
    &+ \sum_{w, \bar{m}_1} p(m_1 \mid w, \bar{a}_1, m_0)p(m_0 \mid w,a_0) \sum_{\bar{a}_1'} Q_Y\left(w,\bar{a}_1', \bar{m}_1\right)p(a_1' \mid w, a_0', m_0) \mathbbm{1}(W=w) \\
    &\qquad \times \left\{\mathbbm{1}(A_0=a_0')-p(a_0'\mid w) \right\} \\
    &+\sum_{w, \bar{m}_1} p(m_1 \mid w, \bar{a}_1, m_0)p(m_0 \mid w, a_0) \sum_{\bar{a}_1'} Q_Y\left(w, \bar{a}_1', \bar{m}_1\right) p(a_1' \mid w, a_0', m_0)p(a_0'\mid w) \left\{\mathbbm{1}(W=w)-p(w) \right\}\\
    =&\frac{\mathbbm{1}(\bar{A}_1=\bar{a}_1)}{p(a_1 \mid W, M_0, a_0)p(a_0\mid W)} \left\{ Q^{\bar{a}}_{M_2}(P)(W, \bar{M}_1) - Q^{\bar{a}}_{M_1}(P)(W, M_0) \right\} \\
    &+\frac{\mathbbm{1}(A_0=a_0)}{p(a_0\mid W)} \left\{Q^{\bar{a}}_{M_1}(P)(W,M_0) - Q^{\bar{a}}_{M_0}(P)(W)\right\}  \\
    & +\frac{p(M_1 \mid W, \bar{a}_1, M_0)p(M_0 \mid W, a_0)}{p(M_1 \mid W,\bar{A}_1, M_0)p(M_0 \mid W, A_0)} \left\{ Y - Q_Y(W, \bar{A}_1, \bar{M}_1)\right\} \\
    &+\frac{p(M_0 \mid W, a_0)}{p(M_0 \mid W, A_0)} \left\{R_{M_1}^{\bar{a}}(P)(W, \bar{A}_1,M_0) -R_{A_1}^{\bar{a}}(P)(W, \bar{A}_0,M_0) \right\} \\
    &+ Q_{M_0}^{\bar{a}}(P)(W, A_0)  - \Psi^{\bar{a}_T}(P),
\end{align*}
}
where the first equality follows by definition, the second equality from the chain rule, and the third equality from evaluating sums and rearranging. 
The generalization to arbitrary $T$ follows immediately.

\subsection{Large sample properties} \label{app:asymptotic}
The analysis of the large sample properties of the proposed estimators will be based on the following expansion:
\begin{align*}
    \Psi^{\bar{a}_T}(\hat{P}_n)- \Psi^{\bar{a}_T}(P) 
    =& \mathbb{P}_n D^*(O)(P) \\ 
    &-\mathbb{P}_n D^*(O)(\hat{P}_n) \\
    &+ (\mathbb{P}_n-P) \left\{D^*(O)(\hat{P}_n)- D^*(O)(P) \right\} \\
    &+R_2(\hat{P}_n, P),
\end{align*}
where $R_2(\hat{P}_n, P)= \mathbb{E}\left\{ D^*(\hat{P}_n)(O) \right\} +  \Psi^{\bar{a}_T}(P_n)- \Psi^{\bar{a}_T}(P)$ is known as the second-order remainder term. 

The term $\mathbb{P}_n D^*(O)(P)$ is asymptotically normally distributed by the Central Limit Theorem. The term $\mathbb{P}_n D^*(O)(\hat{P}_n)$ is known as the `plug-in bias' or first order bias term. The one-step estimator accounts for this bias by adding it on to the estimator $\Psi^{\bar{a}_T}(\hat{P}_n)$. By construction the TMLE solves the efficient influence curve equation, i.e., $\mathbb{P}_n D^*(O)(\hat{P}_n)=o_p(n^{-1/2})$. Assumption 4 (i) implies that $(\mathbb{P}_n-P) \left\{D^*(O)(\hat{P}_n)- D^*(O)(P) \right\}=o_p(n^{-1/2})$. In the following, we derive the remainder term and we show that Assumption 4 (ii) and (iii) implies $R_2(\hat{P}_n, P)=o_p(n^{-1/2})$. 

\subsubsection{Second-order remainder term}

To derive the second-order remainder term, we need the following helper lemmas. 

\begin{lemma} \label{lemma:remain:D_Y}
\small{
\begin{align*}
     &\mathbb{E}\left\{D^*_Y(\hat{P}_n)(O) \right\}\\
     =&\Psi(P) - \mathbb{E}\left\{ H^{\bar{a}}_T(W, \bar{A}_T, \bar{M}_T)\hat{Q}_{n, Y}(W, \bar{A}_T, \bar{M}_T)\right\}\\
    &+\mathbb{E}\left[\left\{\hat{H}^{\bar{a}}_{n, T}(W, \bar{A}_T, \bar{M}_T) - H^{\bar{a}}_T(W, \bar{A}_T, \bar{M}_T) \right\}\left\{Q_Y(W, \bar{A}_T, \bar{M}_T)-\hat{Q}_{n, Y}(W, \bar{A}_T, \bar{M}_T)\right\}\right].
\end{align*}
   }
\end{lemma}

\begin{proof}
This follows by a straightforward expansion:
\begin{align*}
    &\mathbb{E}\left\{D^*_Y(\hat{P}_n)(O) \right\} \\
     =& \mathbb{E}\left[\hat{H}^{\bar{a}}_{n,T}(W, \bar{A}_T, \bar{M}_T)\left\{Q_Y(W, \bar{A}_T, \bar{M}_T)-\hat{Q}_{n, Y}(W, \bar{A}_T, \bar{M}_T)\right\}\right] \\
     =& \mathbb{E}\left[ H^{\bar{a}}_T(W, \bar{A}_T, \bar{M}_T)\left\{Q_Y(W, \bar{A}_T, \bar{M}_T)-\hat{Q}_{n, Y}(W, \bar{A}_T, \bar{M}_T)\right\}\right]\\
    &+\mathbb{E}\left[\left\{\hat{H}^{\bar{a}}_{n, T}(W, \bar{A}_T, \bar{M}_T) - H^{\bar{a}}_T(W, \bar{A}_T, \bar{M}_T) \right\}\left\{Q_Y(W, \bar{A}_T, \bar{M}_T)-\hat{Q}_{n, Y}(W, \bar{A}_T, \bar{M}_T)\right\}\right] \\
    =& \Psi(P) - \mathbb{E}\left\{ H^{\bar{a}}_T(W, \bar{A}_T, \bar{M}_T)\hat{Q}_{n, Y}(W, \bar{A}_T, \bar{M}_T)\right\}\\
    &+\mathbb{E}\left[\left\{\hat{H}^{\bar{a}}_{n, T}(W, \bar{A}_T, \bar{M}_T) - H^{\bar{a}}_T(W, \bar{A}_T, \bar{M}_T) \right\}\left\{Q_Y(W, \bar{A}_T, \bar{M}_T)-\hat{Q}_{n, Y}(W, \bar{A}_T, \bar{M}_T)\right\}\right].
\end{align*}

\end{proof}

\begin{lemma} \label{lemma:remain:D_M}

       \begin{align*}
        &\mathbb{E}\left\{\sum_{t=0}^T D^*_{M_t}(\hat{P}_n)(O) \right\} \\
        =&\mathbb{E}\left[\sum_{t=0}^T \left\{ \hat{W}_{n,t}^{\bar{a}}(W, \bar{A}_t, \bar{M}_{t-1})-W_t^{\bar{a}}(W, \bar{A}_t, \bar{M}_{t-1}) \right\} \left\{\tilde{Q}^{\bar{a}}_{M_t}(W, \bar{M}_{t-1})-\hat{Q}^{\bar{a}}_{n,M_t}(W, \bar{M}_{t-1})\right\}\right]\\
        &+\mathbb{E}\left\{W_T^{\bar{a}}(W, \bar{A}_T, \bar{M}_{T-1}) \hat{Q}_{n, M_{T+1}}^{\bar{a}}(W, \bar{M}_T) \right\} - \mathbb{E}\left\{\hat{Q}_{n,M_0}^{\bar{a}}(W) \right\}.
    \end{align*}
 
\end{lemma}

\begin{proof}

We have
       \begin{align*}
        &\mathbb{E}\left\{\sum_{t=0}^T D^*_{M_t}(\hat{P}_n)(O) \right\} \\
         =&\mathbb{E}\left[\sum_{t=0}^T  \hat{W}_{n,t}^{\bar{a}}(W, \bar{A}_t, \bar{M}_{t-1}) \left\{\tilde{Q}^{\bar{a}}_{M_t}(W, \bar{M}_{t-1})-\hat{Q}^{\bar{a}}_{n,M_t}(W, \bar{M}_{t-1})\right\}\right]\\
        =&\mathbb{E}\left[\sum_{t=0}^T \left\{ \hat{W}_{n,t}^{\bar{a}}(W, \bar{A}_t, \bar{M}_{t-1})-W_t^{\bar{a}}(W, \bar{A}_t, \bar{M}_{t-1}) \right\} \left\{\tilde{Q}^{\bar{a}}_{M_t}(W, \bar{M}_{t-1})-\hat{Q}^{\bar{a}}_{n,M_t}(W, \bar{M}_{t-1})\right\}\right]\\
        &+\mathbb{E}\left[\sum_{t=0}^T W_t^{\bar{a}}(W, \bar{A}_t, \bar{M}_{t-1}) \left\{\tilde{Q}^{\bar{a}}_{M_t}(W, \bar{M}_{t-1})-\hat{Q}^{\bar{a}}_{n,M_t}(W, \bar{M}_{t-1})\right\}\right],
    \end{align*}
where the first equality follows by iterated expectations and the second by a simple expansion. The last term in the display above is a telescoping sum and may be written as follows.
\begin{align*}
    &\mathbb{E}\left[\sum_{t=0}^T W_t^{\bar{a}}(W, \bar{A}_t, \bar{M}_{t-1}) \left\{\tilde{Q}^{\bar{a}}_{M_t}(W, \bar{M}_{t-1})-\hat{Q}^{\bar{a}}_{n,M_t}(W, \bar{M}_{t-1})\right\}\right]\\
    =& \mathbb{E}\left\{W_T^{\bar{a}}(W, \bar{A}_T, \bar{M}_{T-1})\hat{Q}_{n, M_{T+1}}^{\bar{a}}(W, \bar{M}_T)\right\} - \mathbb{E}\left\{\hat{Q}_{n,M_0}^{\bar{a}}(W) \right\}.
\end{align*}

\end{proof}

\begin{lemma} \label{lemma:remain:D_A}

\small{
\begin{align*}
        &\mathbb{E}\left\{\sum_{t=0}^T D^*_{A_t}(\hat{P}_n)(O) \right\} \\
=&\mathbb{E}\Biggl[ \sum_{t=1}^T \left\{\hat{H}^{\bar{a}}_{n,t-1}(W, \bar{A}_{t-1}, \bar{M}_{t-1}) -H^{\bar{a}}_{t-1}(W, \bar{A}_{t-1}, \bar{M}_{t-1})\right\}\sum_{a'_t}  \hat{R}^{\bar{a}}_{n,M_t}(W, a'_t, \bar{A}_{t-1}, \bar{M}_{t-1})\\
    & \qquad \times \left\{ \pi_t(a'_t \mid W, \bar{M}_{t-1}, \bar{A}_{t-1} ) - \hat{\pi}_{n,t}(a'_t \mid W, \bar{M}_{t-1}, \bar{A}_{t-1} )\right\}\Biggr] \\
    &+\mathbb{E}\Bigg[ \sum_{t=0}^T  W_{t-1}^{\bar{a}}(W, \bar{A}_{t-1}, \bar{M}_{t-2}) \sum_{\bar{a}_t'}  \left\{\tilde{R}^{\bar{a}}_{M_t}(W, \bar{a}'_t,  \bar{M}_{t-1})-\hat{R}^{\bar{a}}_{n,M_t}(W, \bar{a}'_t, \bar{M}_{t-1}) \right\}  \\
     &\qquad \times \left\{\prod_{k=0}^t  \hat{\pi}_{n,k}(a'_k \mid W, \bar{M}_{k-1}, \bar{a}'_{k-1} )- \prod_{k=0}^t\pi_k(a'_k \mid W, \bar{M}_{k-1}, \bar{a}_{k-1}' ) \right\}\Bigg]  \\
    &+\mathbb{E}\left\{H^{\bar{a}}_T(W, \bar{A}_T, \bar{M}_T) \hat{Q}_{n,Y}(W, \bar{A}_T, \bar{M}_T) \right\} - \mathbb{E} \left\{W_T^{\bar{a}}(W, \bar{A}_T, \bar{M}_{t-1})\hat{Q}_{n, M_{T+1}}(W, \bar{A}_T, \bar{M}_T)\right\}.
\end{align*}
}
\end{lemma}
\begin{proof}
We have by straightforward expansion
\begin{align*}
        &\mathbb{E}\left\{\sum_{t=0}^T D^*_{A_t}(\hat{P}_n)(O) \right\} \\
        =& \mathbb{E}\left[\sum_{t=1}^T \left\{ \hat{H}^{\bar{a}}_{n,t-1}(W, \bar{A}_{t-1}, \bar{M}_{t-1}) -  H^{\bar{a}}_{t-1}(W, \bar{A}_{t-1}, \bar{M}_{t-1})\right\}\left\{ \hat{R}^{\bar{a}}_{n,M_t}(W, \bar{A}_t, \bar{M}_{t-1}) - \hat{R}^{\bar{a}}_{n,A_t}(W, \bar{A}_{t-1}, \bar{M}_{t-1})\right\} \right] \\
                &+ \mathbb{E}\left[\sum_{t=0}^T H^{\bar{a}}_{t-1}(W, \bar{A}_{t-1}, \bar{M}_{t-1}) \left\{ \hat{R}^{\bar{a}}_{n,M_t}(W, \bar{A}_t, \bar{M}_{t-1}) - \hat{R}^{\bar{a}}_{n,A_t}(W, \bar{A}_{t-1}, \bar{M}_{t-1})\right\} \right]. 
\end{align*}
The first term on the right hand side of the display above may be written as
\begin{align*}
    &\mathbb{E}\left[\sum_{t=1}^T \left\{ \hat{H}^{\bar{a}}_{n,t-1}(W, \bar{A}_{t-1}, \bar{M}_{t-1}) -  H^{\bar{a}}_{t-1}(W, \bar{A}_{t-1}, \bar{M}_{t-1})\right\}\left\{ \hat{R}^{\bar{a}}_{n,M_t}(W, \bar{A}_t, \bar{M}_{t-1}) - \hat{R}^{\bar{a}}_{n,A_t}(W, \bar{A}_{t-1}, \bar{M}_{t-1})\right\} \right] \\
    =&\mathbb{E}\Biggl[ \sum_{t=1}^T \left\{\hat{H}^{\bar{a}}_{n,t-1}(W, \bar{A}_{t-1}, \bar{M}_{t-1}) -H^{\bar{a}}_{t-1}(W, \bar{A}_{t-1}, \bar{M}_{t-1})\right\}\sum_{a'_t}  \hat{R}^{\bar{a}}_{n,M_t}(W, a'_t, \bar{A}_{t-1}, \bar{M}_{t-1})\\
    & \qquad \times \left\{ \pi_t(a'_t \mid W, \bar{M}_{t-1}, \bar{A}_{t-1} ) - \hat{\pi}_{n,t}(a'_t \mid W, \bar{M}_{t-1}, \bar{A}_{t-1} )\right\}\Biggr],
\end{align*}
and the second term as
\begin{align*}
& \mathbb{E}\left[\sum_{t=0}^T H^{\bar{a}}_{t-1}(W, \bar{A}_{t-1}, \bar{M}_{t-1}) \left\{ \hat{R}^{\bar{a}}_{n,M_t}(W, \bar{A}_t, \bar{M}_{t-1}) - \hat{R}^{\bar{a}}_{n,A_t}(W, \bar{A}_{t-1}, \bar{M}_{t-1})\right\} \right] \\
        =&\mathbb{E}\left\{H^{\bar{a}}_{T-1}(W, \bar{A}_{T-1}, \bar{M}_{T-1}) \hat{R}^{\bar{a}}_{n,M_T}(W, \bar{A}_T, \bar{M}_{T-1})  \right\} \\
        &+\mathbb{E}\left[\sum_{t=0}^{T-1} H^{\bar{a}}_{t-1}(W, \bar{A}_{t-1}, \bar{M}_{t-1}) \left\{ \hat{R}^{\bar{a}}_{n,M_t}(W, \bar{A}_t, \bar{M}_{t-1}) - \frac{g_t(M_t \mid W, \bar{a}_t, \bar{M}_{t-1})}{g_t(M_t \mid W, \bar{A}_t, \bar{M}_{t-1})}\hat{R}^{\bar{a}}_{n,A_{t+1}}(W, \bar{A}_t, \bar{M}_t)\right\} \right] \\
        &-\mathbb{E}\left\{\hat{R}_{n, A_0}(W) \right\} \\
=&\mathbb{E}\left[ \sum_{t=0}^T H^{\bar{a}}_{t-1}(W, \bar{A}_{t-1}, \bar{M}_{t-1})  \left\{\hat{R}^{\bar{a}}_{n,M_t}(W, \bar{A}_t, \bar{M}_{t-1})-\tilde{R}^{\bar{a}}_{M_t}(W, \bar{A}_t, \bar{M}_{t-1})\right\} \right]  \\
&+\mathbb{E}\left\{H^{\bar{a}}_T(W, \bar{A}_T, \bar{M}_T) \hat{Q}_{n,Y}(W, \bar{A}_T, \bar{M}_T) \right\} - \mathbb{E}\left\{\hat{R}_{n, A_0}(W) \right\}.
\end{align*}

We can rewrite the first term on the right hand side of the display above as
\begin{align*}
&\mathbb{E}\left[ \sum_{t=0}^T H^{\bar{a}}_{t-1}(W, \bar{A}_{t-1}, \bar{M}_{t-1})  \left\{\hat{R}^{\bar{a}}_{n,M_t}(W, \bar{A}_t, \bar{M}_{t-1})-\tilde{R}^{\bar{a}}_{M_t}(W, \bar{A}_t, \bar{M}_{t-1})\right\} \right] \\
=&      \mathbb{E} \Bigg[\sum_{t=0}^T W_{t-1}^{\bar{a}}(W, \bar{A}_{t-1}, \bar{M}_{t-2}) \sum_{\bar{a}_t'} \left\{\hat{R}_{n, M_t}(W, \bar{a}_t', \bar{M}_{t-1})-\tilde{R}_{M_t}(W, \bar{a}_t', \bar{M}_{t-1}) \right\} \\
& \qquad \times \prod_{k=0}^t \pi_k(a_k' \mid W, \bar{M}_{k-1}, \bar{a}_{k-1}')\Bigg] \\
=&\mathbb{E}\Bigg[ \sum_{t=0}^T  W_{t-1}^{\bar{a}}(W, \bar{A}_{t-1}, \bar{M}_{t-2}) \sum_{\bar{a}_t'}  \left\{\tilde{R}^{\bar{a}}_{M_t}(W, \bar{a}'_t,  \bar{M}_{t-1})-\hat{R}^{\bar{a}}_{n,M_t}(W, \bar{a}'_t, \bar{M}_{t-1}) \right\}  \\
     &\qquad \times \left\{\prod_{k=0}^t  \hat{\pi}_{n,k}(a'_k \mid W, \bar{M}_{k-1}, \bar{a}'_{k-1} )- \prod_{k=0}^t\pi_k(a'_k \mid W, \bar{M}_{k-1}, \bar{a}_{k-1}' ) \right\}\Bigg]  \\
  &-   \mathbb{E} \Bigg[\sum_{t=0}^T W_{t-1}^{\bar{a}}(W, \bar{A}_{t-1}, \bar{M}_{t-2}) \sum_{\bar{a}_t'} \left\{\tilde{R}_{M_t}(W, \bar{a}_t', \bar{M}_{t-1}) -\hat{R}_{n, M_t}(W, \bar{a}_t', \bar{M}_{t-1})\right\} \\
  & \qquad \times \prod_{k=0}^t \hat{\pi}_{n,k}(a_k' \mid W, \bar{M}_{k-1}, \bar{a}_{k-1}')\Bigg].
\end{align*}
The last term in the display above is a telescoping sum and may be written as follows.
\begin{align*}
         &\mathbb{E} \Bigg[\sum_{t=0}^T W_{t-1}^{\bar{a}}(W, \bar{A}_{t-1}, \bar{M}_{t-2}) \sum_{\bar{a}_t'} \left\{\tilde{R}_{M_t}(W, \bar{a}_t', \bar{M}_{t-1}) -\hat{R}_{n, M_t}(W, \bar{a}_t', \bar{M}_{t-1})\right\} \\
     &\qquad \times\prod_{k=0}^t \hat{\pi}_{n,k}(a_k' \mid W, \bar{M}_{k-1}, \bar{a}_{k-1}')\Bigg] \\
     =& \mathbb{E} \left\{W_T^{\bar{a}}(W, \bar{A}_T, \bar{M}_{t-1})\hat{Q}_{n, M_{T+1}}(W, \bar{A}_T, \bar{M}_T)\right\} -\mathbb{E}\left\{\hat{R}_{n, A_0}(W) \right\} .
\end{align*}
Combining the terms above gives the desired expression. 
\end{proof}

\begin{theorem}[Remainder term]\label{theorem:remain}
    We have
    \small{
\begin{align*}
    &R_2(\hat{P}_n, P) \\
         =&\mathbb{E}\left[\left\{\hat{H}^{\bar{a}}_{n, T}(W, \bar{A}_T, \bar{M}_T) - H^{\bar{a}}_T(W, \bar{A}_T, \bar{M}_T) \right\}\left\{Q_Y(W, \bar{A}_T, \bar{M}_T)-\hat{Q}_{n, Y}(W, \bar{A}_T, \bar{M}_T)\right\}\right] \\
    &+   \mathbb{E}\left[ \sum_{t=0}^T \left\{\hat{W}^{\bar{a}}_{n,t}(W, \bar{A}_t, \bar{M}_{t-1})-W^{\bar{a}}_t(W, \bar{A}_t, \bar{M}_{t-1})\right\} \left\{\tilde{Q}_{M_t}(W, \bar{M}_{t-1})-\hat{Q}_{n, M_t}(W, \bar{M}_{t-1})\right\} \right] \\
     &+\mathbb{E}\Bigg[ \sum_{t=0}^T  W_{t-1}^{\bar{a}}(W, \bar{A}_{t-1}, \bar{M}_{t-2}) \sum_{\bar{a}_t'}  \left\{\tilde{R}^{\bar{a}}_{M_t}(W, \bar{a}'_t,  \bar{M}_{t-1})-\hat{R}^{\bar{a}}_{n,M_t}(W, \bar{a}'_t, \bar{M}_{t-1}) \right\}  \\
     &\qquad \times \left\{\prod_{k=0}^t \pi_k(a'_k \mid W, \bar{M}_{k-1}, \bar{a}_{k-1}' ) - \prod_{k=0}^t\hat{\pi}_{n,k}(a'_k \mid W, \bar{M}_{k-1}, \bar{a}'_{k-1} )\right\}\Bigg]  \\
    &+\mathbb{E}\Biggl[ \sum_{t=1}^T \left\{\hat{H}^{\bar{a}}_{n,t-1}(W, \bar{A}_{t-1}, \bar{M}_{t-1}) -H^{\bar{a}}_{t-1}(W, \bar{A}_{t-1}, \bar{M}_{t-1})\right\}\sum_{a'_t}  \hat{R}^{\bar{a}}_{n,M_t}(W, a'_t, \bar{A}_{t-1}, \bar{M}_{t-1})\\
    & \qquad \times \left\{ \pi_t(a'_t \mid W, \bar{M}_{t-1}, \bar{A}_{t-1} ) - \hat{\pi}_{n,t}(a'_t \mid W, \bar{M}_{t-1}, \bar{A}_{t-1} )\right\}\Biggr]. 
\end{align*}
}
\end{theorem}

\begin{proof}
   This follows by combining Lemmas \ref{lemma:remain:D_Y}-\ref{lemma:remain:D_A}.
\end{proof}

\begin{corollary}
If we estimate the mediator density directly, then 
    \begin{align*}
        &\hat{H}^{\bar{a}}_{n, T}(W, \bar{A}_T, \bar{M}_T) - H^{\bar{a}}_T(W, \bar{A}_T, \bar{M}_T) \\
        =&\frac{1}{\prod_{t=0}^T \hat{g}_{n,t}(M_t \mid W, \bar{A}_t, \bar{M}_{t-1})}\left\{ \prod_{t=0}^T \hat{g}_{n,t}(M_t \mid W, \bar{a}_t, \bar{M}_{t-1})-\prod_{t=0}^T g_t(M_t \mid W, \bar{a}_t, \bar{M}_{t-1})\right\} \\
        &+\prod_{t=0}^T\frac{g_t(M_t \mid W, \bar{a}_t, \bar{M}_{t-1})}{ \hat{g}_{n,t}(M_t \mid W, \bar{A}_t, \bar{M}_{t-1})g_t(M_t \mid W, \bar{A}_t, \bar{M}_{t-1})}\\
        &\qquad \times \left\{ \prod_{t=0}^T g_t(M_t \mid W, \bar{A}_t, \bar{M}_{t-1})-\prod_{t=0}^T \hat{g}_{n,t}(M_t \mid W, \bar{A}_t, \bar{M}_{t-1})\right\}.
    \end{align*}
    \end{corollary}
\begin{corollary}
    If we use the reparametrization of the mediator density ratio, then
    \small{
    \begin{align*}
        &\hat{H}^{\bar{a}}_{n, T}(W, \bar{A}_T, \bar{M}_T) - H^{\bar{a}}_T(W, \bar{A}_T, \bar{M}_T) \\
        =& \left\{ \prod_{t=0}^T  \frac{\hat{\gamma}_{n,t, T}(a_t \mid W, \bar{a}_{t-1}, \bar{M}_T)}{\hat{\gamma}_{n,t, T}(A_t \mid W, \bar{A}_{t-1}, \bar{M}_T)} - \prod_{t=0}^T  \frac{\gamma_{t, T}(a_t \mid W, \bar{a}_{t-1}, \bar{M}_T)}{\gamma_{t, T}(A_t \mid W, \bar{A}_{t-1}, \bar{M}_T)} \right\} \prod_{t=0}^T \frac{\pi_t(A_t \mid W, \bar{A}_{t-1}, \bar{M}_{t-1})}{\pi_t(a_t \mid W, \bar{a}_{t-1}, \bar{M}_{t-1})}\\
        &+ \prod_{t=0}^T \frac{\hat{\gamma}_{n, t, T}(a_t \mid W, \bar{a}_{t-1}, \bar{M}_T)}{\hat{\gamma}_{n,t, T}(A_t \mid W, \bar{A}_{t-1}, \bar{M}_T)} \left\{ \prod_{t=0}^T  \frac{\hat{\pi}_{n, t}(A_t \mid W, \bar{A}_{t-1}, \bar{M}_{t-1})}{\hat{\pi}_{n,t}(a_t \mid W, \bar{a}_{t-1}, \bar{M}_{t-1})} - \prod_{t=0}^T  \frac{\pi_t(A_t \mid W, \bar{A}_{t-1}, \bar{M}_{t-1})}{\pi_t(a_t \mid W, \bar{a}_{t-1}, \bar{M}_{t-1})}\right\} \\
        =&\prod_{t=0}^T \frac{\pi_t(A_t \mid W, \bar{A}_{t-1}, \bar{M}_{t-1})}{\pi_t(a_t \mid W, \bar{a}_{t-1}, \bar{M}_{t-1})} \frac{1}{\hat{\gamma}_{n,t, T}(A_t \mid W, \bar{A}_{t-1}, \bar{M}_T)}\\
        &\qquad \times\left\{ \prod_{t=0}^T  \hat{\gamma}_{n, t, T}(a_t \mid W, \bar{a}_{t-1}, \bar{M}_T) - \prod_{t=0}^T \gamma_{t,T}(a_t \mid W, \bar{a}_{t-1}, \bar{M}_T) \right\} \\
        &+\prod_{t=0}^T  \frac{\pi_t(A_t \mid W, \bar{A}_{t-1}, \bar{M}_{t-1})}{\pi_t(a_t \mid W, \bar{a}_{t-1}, \bar{M}_{t-1})} \frac{\gamma_{t, T}(a_t \mid W, \bar{a}_{t-1}, \bar{M}_T)}{\hat{\gamma}_{n, t, T}(A_t \mid W, \bar{A}_{t-1}, \bar{M}_T)\gamma_{t, T}(A_t \mid W, \bar{A}_{t-1}, \bar{M}_T)} \\
        &\qquad \times\left\{  \prod_{t=0}^T \gamma_{t, T}(A_t \mid W, \bar{A}_{t-1}, \bar{M}_T) -\prod_{t=0}^T \hat{\gamma}_{n,t, T}(A_t \mid W, \bar{A}_{t-1}, \bar{M}_T) \right\} \\
        &+\prod_{t=0}^T \frac{\hat{\gamma}_{n,t, T}(a_t \mid W, \bar{a}_{t-1}, \bar{M}_T)}{\hat{\gamma}_{n, t, T}(A_t \mid W, \bar{A}_{t-1}, \bar{M}_T)} \frac{1}{\hat{\pi}_{n,t}(a_t \mid W, \bar{a}_{t-1}, \bar{M}_{t-1})}\\
        &\qquad \times\left\{ \prod_{t=0}^T  \hat{\pi}_{n,t}(A_t \mid W, \bar{A}_{t-1}, \bar{M}_{t-1}) - \prod_{t=0}^T  \pi_t(A_t \mid W, \bar{A}_{t-1}, \bar{M}_{t-1}) \right\} \\
        &+\prod_{t=0}^T \frac{\hat{\gamma}_{n,t, T}(a_t \mid W, \bar{a}_{t-1}, \bar{M}_T)}{\hat{\gamma}_{n,t, T}(A_t \mid W, \bar{A}_{t-1}, \bar{M}_T)} \frac{\pi_t(A_t \mid W, \bar{A}_{t-1}, \bar{M}_{t-1})}{\hat{\pi}_{n,t}(a_t\mid W, \bar{a}_{t-1}, \bar{M}_{t-1})\pi_t(a_t \mid W, \bar{a}_{t-1}, \bar{M}_{t-1})} \\
        &\qquad \times\left\{  \prod_{t=0}^T \pi_t(a_t \mid W, \bar{a}_{t-1}, \bar{M}_{t-1}) -\prod_{t=0}^T  \hat{\pi}_{n,t}(a_t \mid W, \bar{a}_{t-1}, \bar{M}_{t-1}) \right\} .
    \end{align*}}
\end{corollary}

\section{Estimators}

\subsection{IPW estimators}
In the following, we show how the IPW representations of the F-functional can be used to construct simple reweighted estimators.

Based on the representation $\Psi^{\bar{a}_T}(P)=\mathbb{E} \left\{W_T^{\bar{a}}(W, \bar{A}_T, \bar{M}_{T-1})  \sum_{\bar{a}_T'} Q_{M_{T+1}}^{\bar{a}}(W, \bar{M}_T)\right\}$ , the following estimator can be constructed:
\begin{breakablealgorithm}
\caption{IPW1}\label{alg:IPW1}
\begin{itemize}
    \item[Step 1.] Construct estimators  $\hat{Q}_{n,Y}$ of $Q_Y$ and $\hat{\pi}_{n,t}$ of $\pi_t$ for $t=0,...,T$.
    \item[Step 2.] Compute
    \begin{align*}
    \hat{\psi}_n^{IPW_1} \equiv \mathbb{P}_n \left\{ \hat{W}_{n,T}(W, \bar{A}_T, \bar{M}_{T-1})  \sum_{\bar{a}'_T \in \{0,1\}^T} \hat{Q}_{n,Y}(W, \bar{a}'_T, \bar{M}_T) \prod_{t=0}^T \hat{\pi}_{n,t}(a'_t \mid W, \bar{M}_{t-1}, \bar{a}'_{t-1}) \right\}.
\end{align*}
\end{itemize}
\end{breakablealgorithm}

Based on the representation of the F-functional $\Psi^{\bar{a}_T}(P)= \mathbb{E} \left\{ H_T^{\bar{a}}(W, \bar{A}_T, \bar{M}_T) Q_Y\left(W, \bar{A}_T, \bar{M}_T\right)\right\}$, the following estimator can be constructed:
\begin{breakablealgorithm}
\caption{IPW2}\label{alg:IPW2}
\begin{itemize}
    \item[Step 1.] Construct estimator $\hat{Q}_{n,Y}$ of $Q_Y$.
    \item[\textit{Either}]
    \item[Step 2.] Construct estimators $\hat{g}_{n,t}$ of $g_t$ for $t=0,...,T$.
    \item[Step 3.] Compute
    \begin{align*}
    \hat{\psi}_n^{IPW_{2a}} \equiv \mathbb{P}_n \left\{ \prod_{t=0}^T \frac{\hat{g}_{n,t}(M_t \mid W, \bar{a}_t, \bar{M}_{t-1})}{\hat{g}_{n,t}(M_t \mid W, \bar{A}_t, \bar{M}_{t-1})} \hat{Q}_{n,Y}(W, \bar{A}_T, \bar{M}_T) \right\}.
    \end{align*}
    \item[\textit{Or}]
    \item[Step 2.] Construct estimators $\hat{\gamma}_{n, t,T}$ and $\hat{\pi}_{n,t}$ for $t=0,...,T$.
        \item[Step 3.] Compute
    \begin{align*}
    \hat{\psi}_n^{IPW_{2b}} \equiv \mathbb{P}_n \left\{ \prod_{t=0}^T\frac{\hat{\gamma}_{n, t, T}(a_t \mid W, \bar{a}_{t-1}, \bar{M}_T)}{\hat{\pi}_{n,t}(a_t \mid W, \bar{a}_{t-1}, \bar{M}_{t-1})}\frac{\hat{\pi}_{n,t}(A_t \mid W, \bar{A}_{t-1}, \bar{M}_{t-1})}{\hat{\gamma}_{n,t,T}(A_t \mid W, \bar{A}_{t-1}, \bar{M}_T)}\hat{Q}_{n,Y}(W, \bar{A}_T, \bar{M}_T) \right\}.
    \end{align*}
\end{itemize}
\end{breakablealgorithm}

\subsection{Sequential regression estimators}
We can construct the following sequential regression estimator based on the first nested expectation representation of the F-functional.
  \begin{breakablealgorithm}
\caption{SR1}\label{alg:SR1}
\begin{itemize}
    \item[Step 1.] Construct estimators $\hat{Q}_{n,Y}$ of $Q_Y$ and $\hat{\pi}_{n,t}$ of $\pi_t$ for $t=0,...,T$.
    \item[Step 2.] Compute
    \begin{align*}
        \hat{Q}^{\bar{a}}_{n, M_{t+1}}(W, \bar{M}_T)=\sum_{\bar{a}'_T \in \{0,1\}^T} \hat{Q}_{n,Y}(W, \bar{a}_T', \bar{M}_T)\prod_{t=0}^T \hat{\pi}_{n,t}(a_t' \mid W,  \bar{M}_{t-1},\bar{a}_{t-1}').
    \end{align*}
    \item[Step 3.] Recursively for $t=T,T-1,...,0$ regress $\hat{Q}^{\bar{a}}_{n, M_{t+1}}(W, \bar{M}_t)$ onto $(W, \bar{A}_t, \bar{M}_{t-1})$ and evaluate the fitted function at $\bar{A}_t=\bar{a}_t$ and the observed history $(W, \bar{M}_{t-1})$ to obtain estimates $\hat{Q}^{\bar{a}}_{n, M_t}(W, \bar{M}_{t-1})$ of $Q^{\bar{a}}_{M_t}(W,\bar{M}_{t-1})$. 
    \item[Step 4.] Estimate the target parameter
    \begin{align*}
    \hat{\psi}_n^{SR_1} \equiv   \mathbb{P}_n \left\{\hat{Q}^{\bar{a}}_{n, M_0}(W) \right\} .
\end{align*}
\end{itemize}
\end{breakablealgorithm}
We can construct the following sequential regression estimator based on the second nested expectation representation of the F-functional.
\begin{breakablealgorithm}
\caption{SR2}\label{alg:SR2}
\begin{itemize}
    \item[Step 1.] Construct estimators $\hat{Q}_{n,Y}$ of $Q_Y$ and $\hat{\pi}_{n,t}$ of $\pi_t$ for $t=0,...,T$.
    \item[Step 2.] Recursively for $t=T,T-1,...,0$ 
    \begin{itemize}
        \item[(i)] For all $\bar{a}_t' \in \{0,1\}^t$ regress $R_{A_{t+1},n}^{\bar{a}}(W, \bar{a}_t', \bar{M}_t)$ on $(W, \bar{A}_t, \bar{M}_{t-1})$ and evaluate the fitted function at $\bar{A}_t= \bar{a}_t$ and the observed history $(W, \bar{M}_{t-1})$ to obtain an estimate $\hat{\kappa}_{n,t}^{\bar{a}}(W,\bar{a}'_t, M_{t-1})$ of $\mathbb{E} \left\{ R^{\bar{a}}_{A_{t+1}}(W, \bar{a}_t', \bar{M}_t)  \mid W, \bar{A}_t=\bar{a}_t, \bar{M}_{t-1} \right\}$. 
        Then compute $$\hat{R}_{n, M_t}^{\bar{a}}(W, \bar{A}_t, \bar{M}_{t-1})= \sum_{\bar{a}_t' \in \{0,1\}^t} \mathbbm{1}(\bar{A}_t=\bar{a}_t') \hat{\kappa}_{n,t}^{\bar{a}}(W,\bar{a}'_t, M_{t-1}).$$
        \item[(iii)] Compute
        \begin{align*}
            &\hat{R}_{n, A_t}^{\bar{a}}(W, \bar{A}_{t-1}, \bar{M}_{t-1}) \\=&  \sum_{\bar{a}_{t-1}' \in \{0,1\}^{t-1}} \mathbbm{1}(\bar{A}_{t-1}=\bar{a}_{t-1}') \Big\{\hat{\pi}_{n, t}(1 \mid W, \bar{a}_{t-1}', \bar{M}_{t-1}) \hat{\kappa}_{n,t}^{\bar{a}}(W, 1, \bar{a}'_{t-1}, M_{t-1}) \\
            &+ \hat{\pi}_{n, t}(0 \mid W, \bar{a}_{t-1}', \bar{M}_{t-1}) \hat{\kappa}_{n,t}^{\bar{a}}(W, 0, \bar{a}'_{t-1}, M_{t-1})\Big\}.
        \end{align*}
    \end{itemize}
    \item[Step 3.] Estimate the target parameter
    \begin{align*}
    \hat{\psi}_n^{SR_2} \equiv   \mathbb{P}_n \left\{\hat{R}^{\bar{a}}_{n, A_0}(W) \right\} .
\end{align*}
\end{itemize}
\end{breakablealgorithm}

\subsection{TMLE}
The TMLE algorithm can be described as follows. 
\begin{breakablealgorithm}
\caption{TMLE.}\label{alg:tmle1}
\begin{itemize}
    \item[Step 1.] Construct initial estimators $\hat{H}_{n, t}$ and $\hat{\pi}_{n, t}$ of the mediator density ratios $H_t$ and the propensity scores $\pi_t$ for $t=0,...,T$.  Construct initial estimator \(\hat{Q}_{n, Y} (W, \bar{a}_t', \bar{M}_T)\) for
  \(Q_Y (W, \bar{a}_t', \bar{M}_T)\). 
    \item[Step 2.] Update using weighted  intercept-only regression with outcome \(Y\),
  offset \(\mathrm{logit}(\hat{Q}_{n, Y} (W, \bar{A}_t, \bar{M}_T))\), and weights equal to
  \( \hat{H}_{n, T} (W, \bar{A}_T, \bar{M}_T) \). Let \(\hat{\eps}_Y\)  denote the estimated intercept, and compute for all $\bar{a}_t' \in \{0,1\}^t$ the updated estimator 
  \begin{align*}
  \hat{Q}^*_{n, Y} (W, \bar{a}'_t, \bar{M}_T) = \mathrm{expit}(\mathrm{logit}(\hat{Q}_{n, Y} (W, \bar{a}'_t, \bar{M}_T)) + \hat{\eps}_Y),    
  \end{align*}
  Define $\hat{R}_{n, A_{T+1}}^{\bar{a}, *}(W, \bar{a}_T', \bar{M}_T) = \hat{Q}^*_{n, Y} (W, \bar{a}_T', \bar{M}_T) $.
    \item[Step 3.] 
    For $t=T,...,0$ in decreasing order
    \begin{itemize}
    \item[(i)] For all $\bar{a}_t' \in \{0,1\}^t$ regress $R_{n,A_{t+1}}^{\bar{a},*}(W, \bar{a}_t', \bar{M}_t)$ on $(W, \bar{A}_t, \bar{M}_{t-1})$ and evaluate the fitted function at $\bar{A}_t= \bar{a}_t$ and the observed history $(W, \bar{M}_{t-1})$ to obtain an estimate $\hat{\kappa}_{n,t}^{\bar{a},*}(W,\bar{a}'_t, M_{t-1})$ of $\mathbb{E} \left\{ R^{\bar{a},*}_{A_{t+1}}(W, \bar{a}_t', \bar{M}_t)  \mid W, \bar{A}_t=\bar{a}_t, \bar{M}_{t-1} \right\}$. For \(a''=0,1\), then compute $$\hat{R}_{n, M_t}^{\bar{a},*}(W, a'', \bar{A}_{t-1}, \bar{M}_{t-1})= \sum_{\bar{a}_t' \in \{0,1\}^t} \mathbbm{1}(\bar{A}_{t-1}=\bar{a}_{t-1}', a''=a_{t}') \hat{\kappa}_{n,t}^{\bar{a}}(W,\bar{a}'_t, M_{t-1}).$$
    \item[(ii)] Use
  \( \hat{R}_{n, M_t}^{\bar{a},*}(P)(W, 1, \bar{A}_{t-1},\bar{M}_{t-1})-
  \hat{R}_{n, M_t}^{\bar{a},*}(P)(W, 0, \bar{A}_{t-1},\bar{M}_{t-1}) \) as a
  clever covariate in a weighted regression with outcome \(A_t\),
  offset \(\mathrm{logit}(\hat{\pi}_{n, t} ( 1 \mid W, \bar{M}_{t-1},\bar{A}_{t-1}))\), and weights equal to
  \(\hat{H}_{n, t-1} (W, \bar{A}_{t-1}, \bar{M}_{t-1}) \). Let \(\hat{\eps}_{\pi_t}\) denote he estimated coefficient, and compute for all $\bar{a}_t' \in \{0,1\}^t$ the updated estimator 
\begin{align*}
&\hat{\pi}^*_{n, t} (a_t' \mid W, \bar{M}_t , \bar{a}'_{t-1}) = \mathrm{expit}(\mathrm{logit}(\hat{\pi}_{n, t} (a_t' \mid W, \bar{M}_t, \bar{a}'_{t-1}))  \\
&\qquad\qquad + \, \hat{\eps}_{\pi_t} (
\hat{R}_{n, M_t}^{\bar{a},*}(P)(W, 1, \bar{a}'_{t-1},\bar{M}_{t-1}) -
  \hat{R}_{n, M_t}^{\bar{a},*}(P)(W, 0, \bar{a}'_{t-1},\bar{M}_{t-1})))
\end{align*}
  \item[(iii)] Compute for all $\bar{a}_t' \in \{0,1\}^t$ 
        \begin{align*}
            \hat{R}_{n, A_t}^{\bar{a},*}(W, \bar{a}_{t-1}', \bar{M}_{t-1}) 
            & =  \hat{\pi}^*_{n, t}(1 \mid W, \bar{a}_{t-1}', \bar{M}_{t-1}) \hat{\kappa}_{n,t}^{\bar{a},*}(W, 1, \bar{a}'_{t-1}, M_{t-1}) \\
            & \qquad + \, \hat{\pi}^*_{n, t}(0 \mid W, \bar{a}_{t-1}', \bar{M}_{t-1}) \hat{\kappa}_{n,t}^{\bar{a},*}(W, 0, \bar{a}'_{t-1}, M_{t-1}).
        \end{align*}
    \end{itemize}
    \item[Step 4. ] Compute
  \begin{align*}
    \hat{Q}^{\bar{a},*}_{n,M_{T+1}}(P)(W, \bar{M}_T)
    = \sum_{\bar{a}'_T \in \lbrace 0,1\rbrace^T} \hat{Q}^*_{n, Y} (W, \bar{a}_T', \bar{M}_T) \prod_{t=0}^T
    \hat{\pi}^*_{n, t} (a_t' \mid W, \bar{M}_{t-1}, \bar{a}'_{t-1}).
  \end{align*}
  \item[Step 5.] For $t=T,...,0$ in decreasing order
  \begin{itemize}
      \item[(i)] Regress $\hat{Q}^{\bar{a},*}_{n,M_{t+1}}(P)(W, \bar{M}_t)$ on $(W, \bar{A}_t, \bar{M}_{t-1})$ and evaluate the fitted function at $\bar{A}_t= \bar{a}_t$ and the observed history $(W, \bar{M}_{t-1})$ to obtain an estimate $\hat{Q}^{\bar{a}}_{n, M_t}(P)(W, \bar{M}_{t-1})$ of $Q^{\bar{a}}_{M_t}(P)(W, \bar{M}_{t-1})$. 
      \item[(ii)] Update \(\hat{Q}^{\bar{a}}_{M_t}(P)(W, \bar{M}_{t-1})\) in weighted intercept-only regression with outcome \(\hat{Q}^{\bar{a},*}_{n,M_{t+1}}(P)(W, \bar{M}_t)\), offset \(\mathrm{logit}(\hat{Q}^{\bar{a}}_{n, M_t}(P)(W, \bar{M}_{t-1}))\), and weights equal to \(\hat{W}_t^{\bar{a},*} (W,\bar{A}_{t},\bar{M}_{t-1}) = \mathbbm{1}(\bar{A}_t=\bar{a}_t) / (\prod_{l=0}^t\hat{\pi}^*_{n,l} (a_l \mid W, \bar{M}_{l-1}, \bar{a}_{k-1}))\). Let \(\hat{\eps}_{M_t}\)  denote the estimated intercept, and compute for all $\bar{a}_t' \in \{0,1\}^t$ the updated estimator 
  \begin{align*}
  \hat{Q}^{\bar{a},*}_{n, M_t} (W, \bar{a}'_t, \bar{M}_T) = \mathrm{expit}(\mathrm{logit}(\hat{Q}^{\bar{a},*}_{n, M_t} (W, \bar{a}'_t, \bar{M}_T)) + \hat{\eps}_{M_t}).    
  \end{align*}
  \end{itemize}
  \item[Step 6.] Estimate the target parameter 
  \begin{align*}
      \hat{\psi}_n^{TMLE} = \mathbb{P}_n\left\{Q^{\bar{a},*}_{n,M_0}(W)\right\}.
  \end{align*}
\end{itemize}
\end{breakablealgorithm}

\subsection{TMLE for binary $M_t$} \label{app:tmle}

When the mediator is binary, the efficient influence function can be rewritten as in Corollary \ref{cor:DM} below.

\begin{corollary} \label{cor:DM} When \(M_t\) is binary, we have
  \begin{multline*}
      D^*_{M_t}(P)(O)=W_t^{\bar{a}}(W, \bar{A}_t, \bar{M}_{t-1})\left\{Q^{\bar{a}}_{M_{t+1}}(P)(W, 1, \bar{M}_{t-1}) - Q^{\bar{a}}_{M_{t+1}}(P)(W, 0, \bar{M}_{t-1}) \right\} \\
      \times \ \left\{ M_{t} - g_t(1 \mid W, \bar{a}_t, \bar{M}_{t-1}) \right\} \ .
  \end{multline*}
\end{corollary}

We can then replace the loss function $\mathcal{L}_{M_t}(Q_{M_{t}})(O)$ with
\begin{align*}
    \mathcal{L}_{M_t}(g_t)(O) = - \left\{ M_t \log  g_t(M_t \mid W, \bar{a}_t', \bar{M}_{t-1}) + \left(1 - M_t \right)\log \left(1 - g_t(M_t \mid W, \bar{a}_t', \bar{M}_{t-1})   \right)\right\} ,
\end{align*}

Moreover, we can replace the least favorable submodel $Q_{M_t, \varepsilon}(W,\bar{M}_{t-1})$ with 
\begin{multline*}
    g_{t, \varepsilon}(M_t \mid W, \bar{a}_t', \bar{M}_{t-1})=\text{expit}\Big(\text{logit} \left(g_t(M_t \mid W, \bar{a}_t', \bar{M}_{t-1})\right) \\
    + \varepsilon  W_t^{\bar{a}}(W, \bar{A}_t, \bar{M}_{t-1})\left\{Q^{\bar{a}}_{M_{t+1}}(P)(W, 1, \bar{M}_{t-1}) - Q^{\bar{a}}_{M_{t+1}}(P)(W, 0, \bar{M}_{t-1}) \right\} \Big).
\end{multline*}

Together, these satisfy
\begin{multline*}
    \frac{d}{d\varepsilon}\Bigg\vert_{\varepsilon=0}  \mathcal{L}_{M_t}(g_{t, \varepsilon})(O) = W_t^{\bar{a}}(W, \bar{A}_t, \bar{M}_{t-1})\left\{Q^{\bar{a}}_{M_{t+1}}(P)(W, 1, \bar{M}_{t-1}) - Q^{\bar{a}}_{M_{t+1}}(P)(W, 0, \bar{M}_{t-1}) \right\} \\
      \times \ \left\{ M_{t} - g_t(1 \mid W, \bar{a}_t, \bar{M}_{t-1}) \right\} \ .
\end{multline*}

Then a targeted minimum loss estimator (TMLE) may be computed in the following steps.
\begin{breakablealgorithm}
\caption{TMLE updating the mediator density (binary \(M_t\))}\label{alg:tmle2}
\begin{itemize}
    \item[Step 1.] Construct initial estimators $\hat{H}^{0}_{n, t}$ and $\hat{\pi}^0_{n, t}$ of the mediator density ratios $H_t$ and the propensity scores $\pi_t$ for $t=0,...,T$. Construct initial estimator \(\hat{Q}^0_{n, Y} (W, \bar{a}_t', \bar{M}_T)\) for
  \(Q_Y (W, \bar{a}_t', \bar{M}_T)\). 
\end{itemize}
Recursively for \(k\ge 1\) until convergence: 
\begin{itemize}
    \item[Step 2.]  Update using weighted  intercept-only regression with outcome \(Y\),
  offset \(\mathrm{logit}(\hat{Q}^{k-1}_{n, Y} (W, \bar{A}_t, \bar{M}_T))\), and weights equal to
  \( \hat{H}^{k-1}_{n, T} (W, \bar{A}_T, \bar{M}_T) \). Let \(\hat{\eps}^k_Y\)  denote the estimated intercept, and compute for all $\bar{a}_t' \in \{0,1\}^t$ the updated estimator 
  \begin{align*}
  \hat{Q}^k_{n, Y} (W, \bar{a}'_t, \bar{M}_T) = \mathrm{expit}(\mathrm{logit}(\hat{Q}^{k-1}_{n, Y} (W, \bar{a}'_t, \bar{M}_T)) + \hat{\eps}^k_Y),    
  \end{align*}
  Define $\hat{R}_{n, A_{T+1}}^{\bar{a}, k}(W, \bar{a}_T', \bar{M}_T) = \hat{Q}^k_{n, Y} (W, \bar{a}_T', \bar{M}_T) $.
    \item[Step 3.] 
    For $t=T,...,0$ in decreasing order
    \begin{itemize}
    \item[(i)] For all $\bar{a}_t' \in \{0,1\}^t$ 
    compute 
    $$\hat{R}_{n, M_t}^{\bar{a},k}(W, \bar{a}_t', \bar{M}_{t-1})= \sum_{m_t=0,1}
    R_{n,A_{t+1}}^{\bar{a},k}(W, \bar{a}_t', m_t, \bar{M}_{t-1}) g^{k-1}_t (m_t \mid W, \bar{a}_t, \bar{M}_{t-1}).$$
    \item[(ii)] Use
  \( \hat{R}_{n, M_t}^{\bar{a},k}(P)(W, 1, \bar{A}_{t-1},\bar{M}_{t-1})-
  \hat{R}_{n, M_t}^{\bar{a},k}(P)(W, 0, \bar{A}_{t-1},\bar{M}_{t-1}) \) as a
  clever covariate in a weighted regression with outcome \(A_t\),
  offset \(\mathrm{logit}(\hat{\pi}^{k-1}_{n, t} ( 1 \mid W, \bar{M}_{t-1},\bar{A}_{t-1}))\), and weights equal to
  \(\hat{H}^{k-1}_{n, t-1} (W, \bar{A}_{t-1}, \bar{M}_{t-1}) \). Let \(\hat{\eps}^k_{\pi_t}\) denote he estimated coefficient, and compute for all $\bar{a}_t' \in \{0,1\}^t$ the updated estimator 
\begin{align*}
&\hat{\pi}^k_{n, t} (a_t' \mid W, \bar{M}_t , \bar{a}'_{t-1}) = \mathrm{expit}(\mathrm{logit}(\hat{\pi}^{k-1}_{n, t} (a_t' \mid W, \bar{M}_t, \bar{a}'_{t-1}))  \\
&\qquad\qquad + \, \hat{\eps}^k_{\pi_t} (
\hat{R}_{n, M_t}^{\bar{a},k}(P)(W, 1, \bar{a}'_{t-1},\bar{M}_{t-1}) -
  \hat{R}_{n, M_t}^{\bar{a},k}(P)(W, 0, \bar{a}'_{t-1},\bar{M}_{t-1})))
\end{align*}
  \item[(iii)] Compute for all $\bar{a}_t' \in \{0,1\}^t$ 
        \begin{align*}
            \hat{R}_{n, A_t}^{\bar{a},k}(W, \bar{a}_{t-1}', \bar{M}_{t-1}) 
            & =  \hat{\pi}^k_{n, t}(1 \mid W, \bar{a}_{t-1}', \bar{M}_{t-1}) \hat{R}_{n,M_t}^{\bar{a},k}(W, 1, \bar{a}'_{t-1}, M_{t-1}) \\
            & \qquad + \, \hat{\pi}^k_{n, t}(0 \mid W, \bar{a}_{t-1}', \bar{M}_{t-1}) \hat{R}_{n,M_t}^{\bar{a},k}(W, 0, \bar{a}'_{t-1}, M_{t-1}).
        \end{align*}
    \end{itemize}
    \item[Step 4. ] Compute
  \begin{align*}
    \hat{Q}^{\bar{a},k}_{n,M_{T+1}}(P)(W, \bar{M}_T)
    = \sum_{\bar{a}'_T \in \lbrace 0,1\rbrace^T} \hat{Q}^k_{n, Y} (W, \bar{a}_T', \bar{M}_T) \prod_{t=0}^T
    \hat{\pi}^k_{n, t} (a_t' \mid W, \bar{M}_{t-1}, \bar{a}'_{t-1}).
  \end{align*}
  \item[Step 5.] For $t=T,...,0$ in decreasing order
  \begin{itemize}
      \item[(i)] Use
  \( \hat{Q}_{n, M_{t+1}}^{\bar{a},k}(P)(W, 1, \bar{M}_{t-1})-
  \hat{Q}_{n, M_{t+1}}^{\bar{a},k}(P)(W, 1, \bar{M}_{t-1}) \) as a
  clever covariate in a weighted regression with outcome \(M_t\),
  offset \(\mathrm{logit}(\hat{g}^{k-1}_{n, t} ( 1 \mid W,\bar{A}_{t}, \bar{M}_{t-1}))\), and weights equal to
  \(\hat{W}_t^{\bar{a},k} (W,\bar{A}_{t},\bar{M}_{t-1}) = \mathbbm{1}(\bar{A}_t=\bar{a}_t) / (\prod_{l=0}^t\hat{\pi}^k_{n,l} (a_l \mid W, \bar{M}_{l-1}, \bar{a}_{k-1}))\). Let \(\hat{\eps}^k_{g_t}\) denote the estimated coefficient, and compute for all $\bar{a}_t' \in \{0,1\}^t$ the updated estimator 
\begin{align*}
&\hat{g}^k_{n, t} (m_t \mid W , \bar{a}'_{t}, \bar{M}_{t-1}) = \mathrm{expit}(\mathrm{logit}(\hat{g}^{k-1}_{n, t} (m_t \mid W, \bar{a}'_{t}, \bar{M}_t))  \\
&\qquad\qquad + \, \hat{\eps}^k_{g_t} (
\hat{Q}_{n, M_{t+1}}^{\bar{a},k}(P)(W, 1, \bar{M}_{t-1})-
  \hat{Q}_{n, M_{t+1}}^{\bar{a},k}(P)(W, 1, \bar{M}_{t-1})).
\end{align*}
  \item[(ii)] Compute an estimator $\hat{Q}^{\bar{a},k}_{n, M_t}(P)(W, \bar{M}_{t-1})$ for $Q^{\bar{a}}_{M_t}(P)(W, \bar{M}_{t-1})$ as
  $$\hat{Q}^{\bar{a},k}_{n, M_t}(P)(W, \bar{M}_{t-1}) = \sum_{m_t =0,1} 
  \hat{Q}^{\bar{a},k}_{n,M_{t+1}}(P)(W, m_t, \bar{M}_{t-}) g_{n,t}^k( m_t \mid W, \bar{a}_t, \bar{M}_{t-1}).
  $$
  \item[Step 6.] Estimate the target parameter 
  \begin{align*}
      \hat{\psi}_n^{TMLE} = \mathbb{P}_n\left\{Q^{\bar{a},k}_{n,M_0}(W)\right\}.
  \end{align*}
\end{itemize}
\end{itemize}
\end{breakablealgorithm}

\section{Alternative interpretations of the longitudinal front-door functional} \label{app:PIIE}
  Here, we generalize the estimands of \cite{fulcher2020robust} and \cite{wen2023causaleffectsinterveningvariables} to the longitudinal setting and show that they are identified by the longitudinal front-door functional. 

\subsection{Identification the PIIE}
Let $Y(\bar{A}_T, \bar{M}_T(\bar{a}_T))$ be the potential outcome under an intervention that for $t=0,...,T$ sets $M_t$ to the value it would have taken when exposure history equals $\bar{a}_t$ while keeping $A_t$ at its natural value. Define
\begin{align*}
    \Phi^{\bar{a}_T}(P)=\mathbb{E}\left\{Y(\bar{A}_T, \bar{M}_T(\bar{a}_T))\right\}.
\end{align*}
The longitudinal version of the population intervention indirect effect (PIIE) of \cite{fulcher2020robust} would be the contrast $\mathbb{E}(Y) -\Phi^{\bar{a}_T}(P)$. Below, we show that $\Phi^{\bar{a}_T}(P)$ is identified by the longitudinal front-door formula.  

\begin{theorem} \label{theorem:PIIE}
Suppose that the following assumptions hold for $t=0,...,T$.
\begin{itemize}
    \item[(i)] Sequential exchangeability:  $$\underbar{M}_t(\bar{a}_T) \independent A_t \mid W, \bar{A}_{t-1}, \bar{M}_{t-1},$$  $$Y(\bar{m}_T) \independent M_t \mid W, \bar{A}_t, \bar{M}_{t-1},$$
    \item[(ii)] Cross-world exchangeability: $Y(\bar{m}_T) \independent \bar{M}_T(\bar{a}_T) \mid W$,
    \item[(iii)] Consistency:  $$\text{If } \bar{A}_t=\bar{a}_t \text{ then } M_t(\bar{a}_t)=M_t \text{ with probability 1},$$ $$\text{If } \bar{M}_t=\bar{m}_t \text{ then } Y(\bar{m}_t)=Y \text{ with probability 1}.$$ 
\end{itemize}
Then $\Phi^{\bar{a}_T}(P)$ is identified by the longitudinal front-door functional. 
\end{theorem}

\begin{proof}
By the law of total probability, we have
        \begin{align*}
        \Phi^{\bar{a}_T}(P)&=\iint \mathbb{E}\left\{ Y(\bar{A}_T, \bar{M}_T(\bar{a}_T)) \mid \bar{M}_T(\bar{a}_T) = \bar{m}_T, W=w \right\} \\
        & \qquad \times \prod_{t=0}^Tp(M_t(\bar{a}_t) = m_t \mid \bar{M}_{t-1}(\bar{a}_{t-1}) = \bar{m}_{t-1},  w)p(w) \diff \bar{m}_T \diff w.
    \end{align*}
By exposure-mediator sequential exchangeability and consistency, we have
\begin{align*}
    p(M_t(\bar{a}_t) = m_t \mid \bar{M}_{t-1}(\bar{a}_{t-1}) = \bar{m}_{t-1},  w) = p(M_t = m_t \mid \bar{A}_t=\bar{a}_T, \bar{M}_{t-1} = \bar{m}_{t-1},  w).
\end{align*}
Moreover, we have
\begin{align*}
    \mathbb{E}\left\{ Y(\bar{A}_T, \bar{M}_T(\bar{a}_T)) \mid \bar{M}_T(\bar{a}_T) = \bar{m}_T, W=w \right\} &=\mathbb{E}\left\{ Y(\bar{m}_T) \mid \bar{M}_T(\bar{a}_T) = \bar{m}_T, W=w \right\} \\
    &=\mathbb{E}\left\{ Y(\bar{m}_T) \mid W=w \right\} \\
    &= \sum_{a_0'} \mathbb{E}\left\{ Y(\bar{m}_T) \mid A_0=a_0', W=w \right\}p(a_0' \mid w) \\
    &=\sum_{a_0'} \mathbb{E}\left\{ Y(\bar{m}_T) \mid M_0=m_0, A_0=a_0', W=w \right\}p(a_0' \mid w) \\
    &=\cdots \\
    &=\sum_{\bar{a}_T'} \mathbb{E}\left\{ Y(\bar{m}_T) \mid \bar{m}_T, \bar{a}_T',w \right\} \prod_{t=0}^T p(a_t' \mid \bar{m}_t, \bar{a}_{t-1}', w) \\
     &=\sum_{\bar{a}_T'} \mathbb{E}( Y\mid \bar{m}_T, \bar{a}_T',w ) \prod_{t=0}^T p(a_t' \mid \bar{m}_t, \bar{a}_{t-1}', w) ,
\end{align*}
    where the first equality follows by definition of the counterfactuals, the second equality by  cross-world exchangeability, the third by the law of total probability, and the fourth by mediator-outcome sequential exchangeability. We repeat the arguments in steps three and four until we arrive at the expression on the right-hand side of the sixth equality. The final equality follows by the consistency assumption.
\end{proof}

\subsection{Identification of interventionist version of PIIE}
Let $A^M_t$ be a hypothetical variable that only affects $Y$ through its effect on $\underline{M}_t$. In the observed data $A_t=A^M_t$ but we assume that in principle $A_t^M$ can be intervened upon. 

Let $Y(\bar{a}_T^M )$ be the potential outcome under an intervention that sets $\bar{A}_T^M=\bar{a}_T^M $, and define
\begin{align*}
    \Omega^{\bar{a}_T^M}(P)= \mathbb{E}\left\{Y(\bar{a}_T^M )\right\}.
\end{align*}
This is the longitudinal version of the estimand proposed by \cite{wen2023causaleffectsinterveningvariables}. It can be used to defined various causal estimands, for example an interventionist version of the PIIE can be defined as the contrast $\mathbb{E}(Y)-\Omega^{\bar{0}_T}(P)$.

Below, we show that $\Omega^{\bar{a}_T^M}(P)$ can be identified by the longitudinal front-door formula. 

    \begin{figure}[ht]
      \centering
      \begin{tikzpicture}
      \tikzset{line width=1.5pt, outer sep=0pt, ell/.style={draw,fill=white, inner sep=2pt, line width=1.5pt}, swig vsplit={gap=5pt, inner line width right=0.5pt}};
      \node[name=A0, ell, shape=ellipse]{$A_0$};
      \node[name=A0M,shape=swig vsplit, right=4mm of A0]{
      \nodepart{left}{$A_0^M$}
      \nodepart{right}{$a_0^M$} };
      \node[name=M0, ell, shape=ellipse,  below right =15mm and .25mm of A0M]{$M_0(a_0^M)$};
      \node[name=A1, ell, shape=ellipse,  right=4mm of A0M]{$A_1$};
      \node[name=A1M,shape=swig vsplit, right=4mm of A1]{
      \nodepart{left}{$A_1^M$}
      \nodepart{right}{$a_1^M$} };
      \node[name=M1, ell, shape=ellipse, below right =15mm and .25mm of A1M]{$M_1(\bar{a}_1^M)$};
      \node[name=Y, ell, shape=ellipse,  right=15mm of M1]{$Y(\bar{a}_1^M)$};
      \draw[-latex,line width=1.5pt,>=stealth] (A0) to (A0M);
      \draw[-latex,line width=1.5pt,>=stealth] (A1) to (A1M);
      \draw[-latex, bend left] (A0) to (A1);
      \draw[-latex]  (M0)--(A1);
      \draw[-latex, color=red] (A0)  to[out=35, in=110] (Y);
      \draw[-latex, color=red] (A1)  to[out=35, in=115] (Y);
      \draw[-latex] (M0)--(M1);
      \draw[-latex, bend right] (M0) to (Y);
      \draw[-latex] (M1)--(Y);
      \draw[-latex] (A0M)--(M0);
      \draw[-latex] (A0M)--(M1);
      \draw[-latex] (A1M)--(M1);
     \node[name=U,draw,dashed,circle, above=20mm of A1]{$\boldsymbol{U}$};
     \draw[-latex] (U)--(A0);
      \draw[-latex]  (U)--(A1);
     \draw[-latex]  (U)--(Y.north);
\end{tikzpicture}
\caption{SWIG corresponding to an intervention that sets $\bar{A}_T^M$ to $\bar{a}_T^M$.}
        \label{fig:SWIG}
  \end{figure}

\begin{theorem}
Suppose that the following assumptions hold for $t=0,...,T$.
\begin{itemize}
    \item[(i)] Determinism: $A_t=A^M_t$,
    \item[(ii)] Dismissible components
    $$Y \independent \bar{A}_T^M \mid \bar{A}_T, \bar{M}_T, W,$$
    $$M_t \independent \bar{A}_t \mid \bar{A}_t^M, \bar{M}_{t-1}, W,$$
    \item[(iii)] Consistency: $$\text{If } \bar{A}_t^M = \bar{a}_t^M \text{ then } Y(\bar{a}_t^M)=Y \text{with probability 1}.$$
\end{itemize}
Then $ \Omega^{\bar{a}_T^M}(P)$ is identified by the longitudinal front-door functional.
    
\end{theorem}
\begin{proof}
\begin{align*}
     \Omega^{\bar{a}_T^M}(P) &=\int \sum_{a_0'} \mathbb{E}\left\{Y(\bar{a}_T^M) \mid A_0=a_0', w \right\} p(a_0' \mid w) p(w) \diff w \\
    &=\int_{w} \sum_{a_0'} \mathbb{E}\left\{Y(\bar{a}_T^M) \mid A_0^M=a_0^M, A_0=a_0', w \right\} p(a_0' \mid w) p(w)  \diff w \\
    &=\iint \sum_{a_0'} \mathbb{E}\left\{Y(\bar{a}_T^M) \mid M_0 = m_0, A_0^M=a_0^M, A_0=a_0',w \right\} \\
    & \qquad \times p(m_0 \mid A_0^M=a_0^M, A_0=a_0', w) p(a_0' \mid w) p(w)  \diff m_0 \diff w \\
    &=\iint \sum_{a_0'} \mathbb{E}\left\{Y(\bar{a}_T^M) \mid M_0 = m_0, A_0=a_0',w \right\} \\
    & \qquad \times p(m_0 \mid A_0^M=a_0^M, A_0=a_0', w) p(a_0' \mid w) p(w)  \diff m_0 \diff w \\
    &= \cdots \\
    &=\iint \sum_{\bar{a}_T'} \mathbb{E}\left\{Y(\bar{a}_T^M) \mid \bar{m}_T, \bar{A}_T^M=\bar{a}_T^M, \bar{A}_T=\bar{a}_T', w \right\}  \\
    & \qquad \times \prod_{t=0}^T p(m_t \mid \bar{m}_{t-1}, \bar{A}_t^M=\bar{a}_t^M, \bar{A}_t=\bar{a}_t',  w)  p(a_t' \mid \bar{m}_t, \bar{a}_{t-1}', w) p(w) \diff \bar{m}_T \diff w \\
    &=\iint \sum_{\bar{a}_T'} \mathbb{E}(Y\mid \bar{m}_T, \bar{A}_T^M=\bar{a}_T^M, \bar{A}_T=\bar{a}_T', w )  \\
    & \qquad \times \prod_{t=0}^T p(m_t \mid \bar{m}_{t-1}, \bar{A}_t^M=\bar{a}_t^M, \bar{A}_t=\bar{a}_t',  w)  p(a_t' \mid \bar{m}_t, \bar{a}_{t-1}', w) p(w) \diff \bar{m}_T \diff w \\
    &=\iint \prod_{t=0}^Tp(m_t \mid \bar{m}_{t-1}, \bar{A}_t^M=\bar{a}_t^M, \bar{A}_t=\bar{a}_t^M,  w) \\
    & \qquad \times \sum_{\bar{a}_T'} \mathbb{E}(Y \mid \bar{m}_T, \bar{A}_T^M=\bar{a}_T', \bar{A}_T=\bar{a}_T',  w )   \prod_{t=0}^T  p(a_t' \mid \bar{m}_t, \bar{a}_{t-1}', w) p(w) \diff \bar{m}_T \diff w \\
    &=\iint \prod_{t=0}^Tp(m_t \mid \bar{m}_{t-1}, \bar{A}_t=\bar{a}_t^M,  w) \\
    & \qquad \times \sum_{\bar{a}_T'} \mathbb{E}(Y \mid \bar{m}_T, \bar{A}_T=\bar{a}_T',  w )   \prod_{t=0}^T  p(a_t' \mid \bar{m}_t, \bar{a}_{t-1}', w) p(w) \diff \bar{m}_T \diff w, 
\end{align*}
where the first equality follow from the law of total probability, the second equality follows from $Y(\bar{a}_T^M) \independent \bar{A}_t^M \mid \bar{A}_t, W$ which can be read off the SWIG in Figure \ref{fig:SWIG}, the third by the law of total probability, and the fourth by applying $Y(\bar{a}_T^M) \independent \bar{A}_t^M \mid \bar{A}_t, W$ again. We then repeat steps one through four until we arrive at the expression after the sixth equality. The seventh equality follows by consistency, and the eight by the dismissible components. The final equality follows by determinism. 
\end{proof}

\section{Details on simulations}

\subsection{Data generating distributions}
Data are generating from the following distributions. 
\begin{align*}
    U &\sim Bern(0.5), \\
    L_{0,1} &\sim  Bern(0.6), \\
    L_{0,2} &\sim Bern\left\{\text{expit}\left(1-L_{0,1} \right) \right\}, \\
    A_t &\sim Bern\big\{\text{expit}\big(-1 + 0.8 L_{0,1} - 0.5L_{0,2} + L_{0,1}L_{0,2} + 2U \\ &\qquad \qquad \qquad \quad + \mathbbm{1}(t>1)(0.4A_{t-1} + 0.25M_{t-1})\big\}, \\
   M_t &\sim  \begin{cases}
    Bern\big\{\text{expit}\big(-1 + L_{0,1} - 0.75L_{0,2} + 0.75L_{0,1}L_{0,2} +  A_t   \\ \qquad \qquad \qquad + \mathbbm{1}(t>1)(0.5A_{t-1} + 0.252M_{t-1}) \big\}, \text{ (settings 1 and 2)}\\
    N\big\{1+ L_{0,1} - 0.75L_{0,2} + 0.75L_{0,1}L_{0,2} +  A_t  \\ \qquad + \mathbbm{1}(t>1)(0.5 A_{t-1} + 0.25M_{t-1}), 1\big\}, \text{ (settings 3 and 4)}
    \end{cases} \\
    Y &\sim \begin{cases} Bern\big\{\text{expit}\big(1 - L_{0,1} + 0.75L_{0,2} - L_{0,1}L_{0,2} - U \\  \qquad \qquad \qquad - 1.5M_T - M_{T-1}\big) \big\}, \text{ (setting 1)}\\
     Bern\big\{\text{expit}\big(1 - L_{0,1} + 0.75L_{0,2} - L_{0,1}L_{0,2} - U \\  \qquad \qquad \qquad  - 0.75M_T - 0.5M_{T-1} - 0.75A_T - 0.5A_{T-1}\big) \big\}, \text{(setting 2)} \\
     Bern\big\{\text{expit}\big(1 - L_{0,1} + 0.75L_{0,2} - L_{0,1}L_{0,2} - U \\  \qquad \qquad \qquad - 0.5 M_T - 0.3M_{T-1}\big) \big\}, \text{(setting 3)}\\
     Bern\big\{\text{expit}\big(1 - L_{0,1} + 0.75L_{0,2} - L_{0,1}L_{0,2} - U \\  \qquad \qquad \qquad  - 0.25M_T - 0.15M_{T-1} - 0.25A_T - 0.15A_{T-1}\big) \big\}, \text{(setting 4)}\\
    \end{cases}
\end{align*}
for $T=1$. To approximate the true value of $\Psi^{\bar{a}_T}(P)$ we generated a large data set of size $10^7$ based on the distributions above with $\bar{A}_T$ set to $\bar{a}_T$ in the distribution of $M_t$, and took the sample mean of $Y$.  With $T=1$ the approximations of the true values are as follows. DGM (1) $\Psi^{\bar{1}_T}(P) \approx 0.249$, $\Psi^{\bar{0}_T}(P) \approx 0.353$. DGM (2) $\Psi^{\bar{1}_T}(P) \approx 0.252$, $\Psi^{\bar{0}_T}(P) \approx 0.301$. DGM (3) $\Psi^{\bar{1}_T}(P) \approx 0.178$, $\Psi^{\bar{0}_T}(P) \approx 0.310$. DGM (4) $\Psi^{\bar{1}_T}(P) \approx 0.261$, $\Psi^{\bar{0}_T}(P) \approx 0.348$.

\subsection{Additional results}
Below we report the simulation results for DGM (2)-(4).
\begin{figure}
    \centering
    \includegraphics[width=0.9\textwidth]{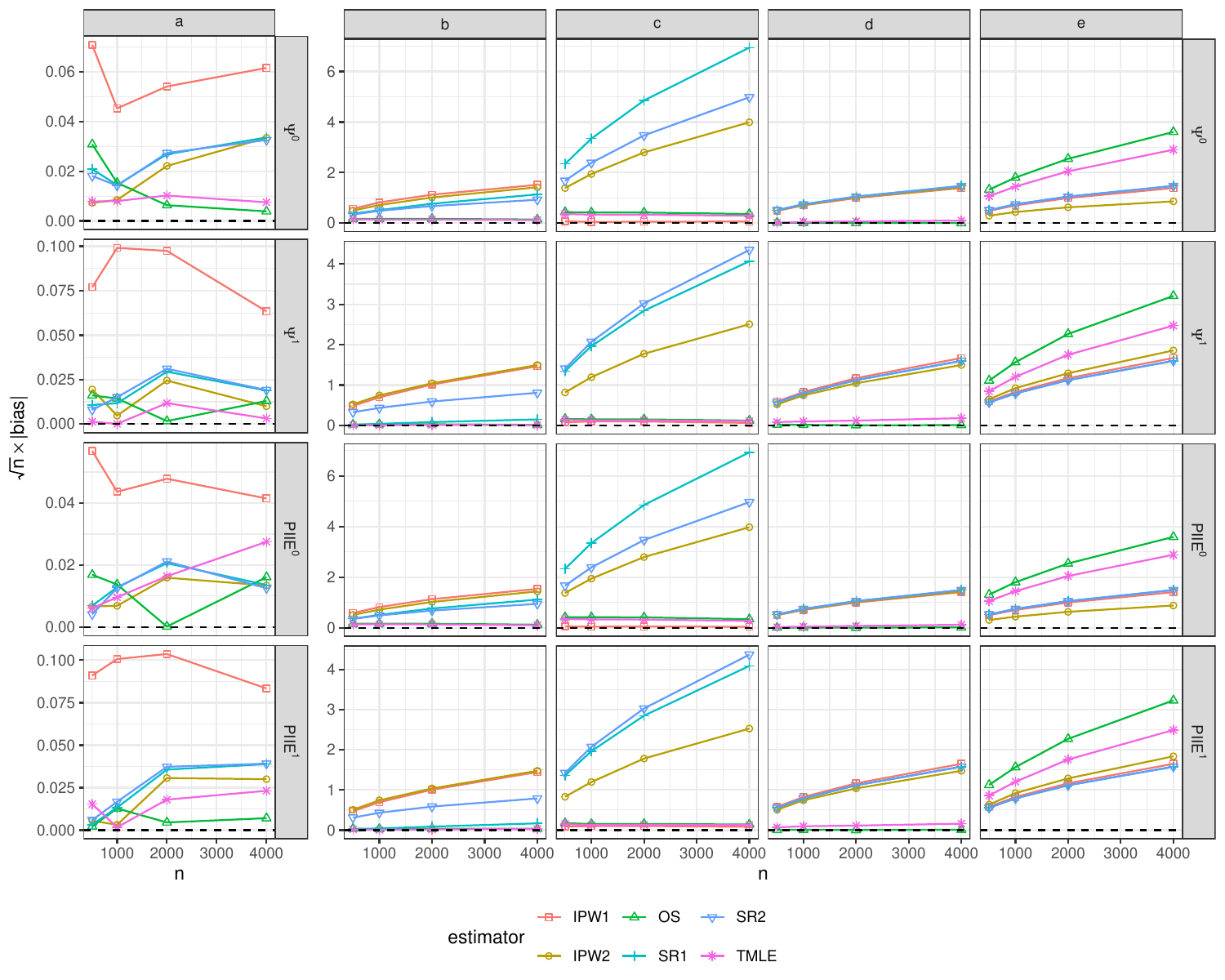}
    \caption{Simulation results for DGM (2). Absolute bias scaled by $\sqrt{n}$.}
    \label{fig:sim_bias2}
\end{figure}
\begin{figure}[H]
    \centering
    \includegraphics[width=0.9\textwidth]{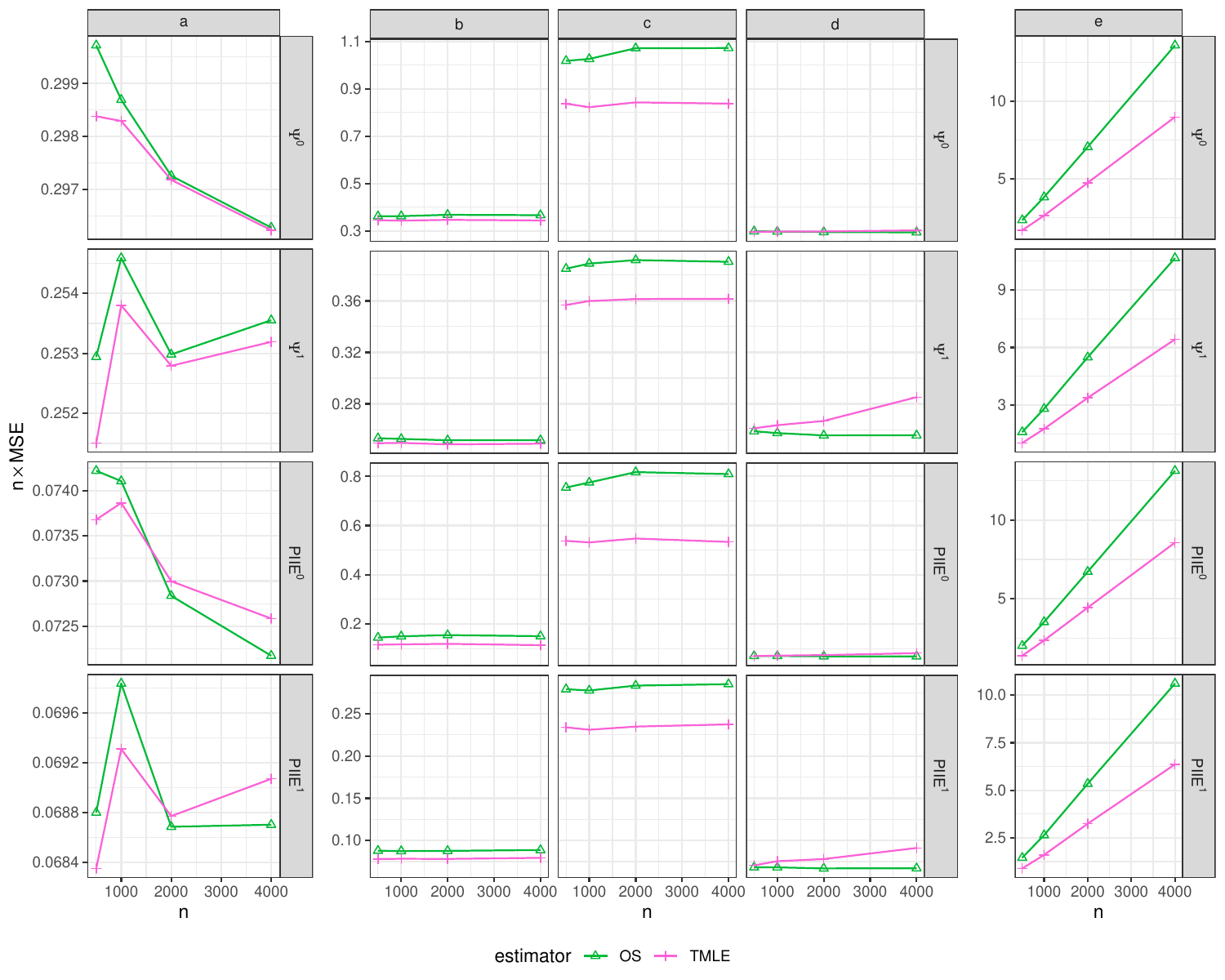}
    \caption{Simulation results for DGM (2). $\mbox{MSE}$ scaled by $n$.}
    \label{fig:sim_mse2}
\end{figure}

\begin{figure}[H]
    \centering
    \includegraphics[width=0.9\textwidth]{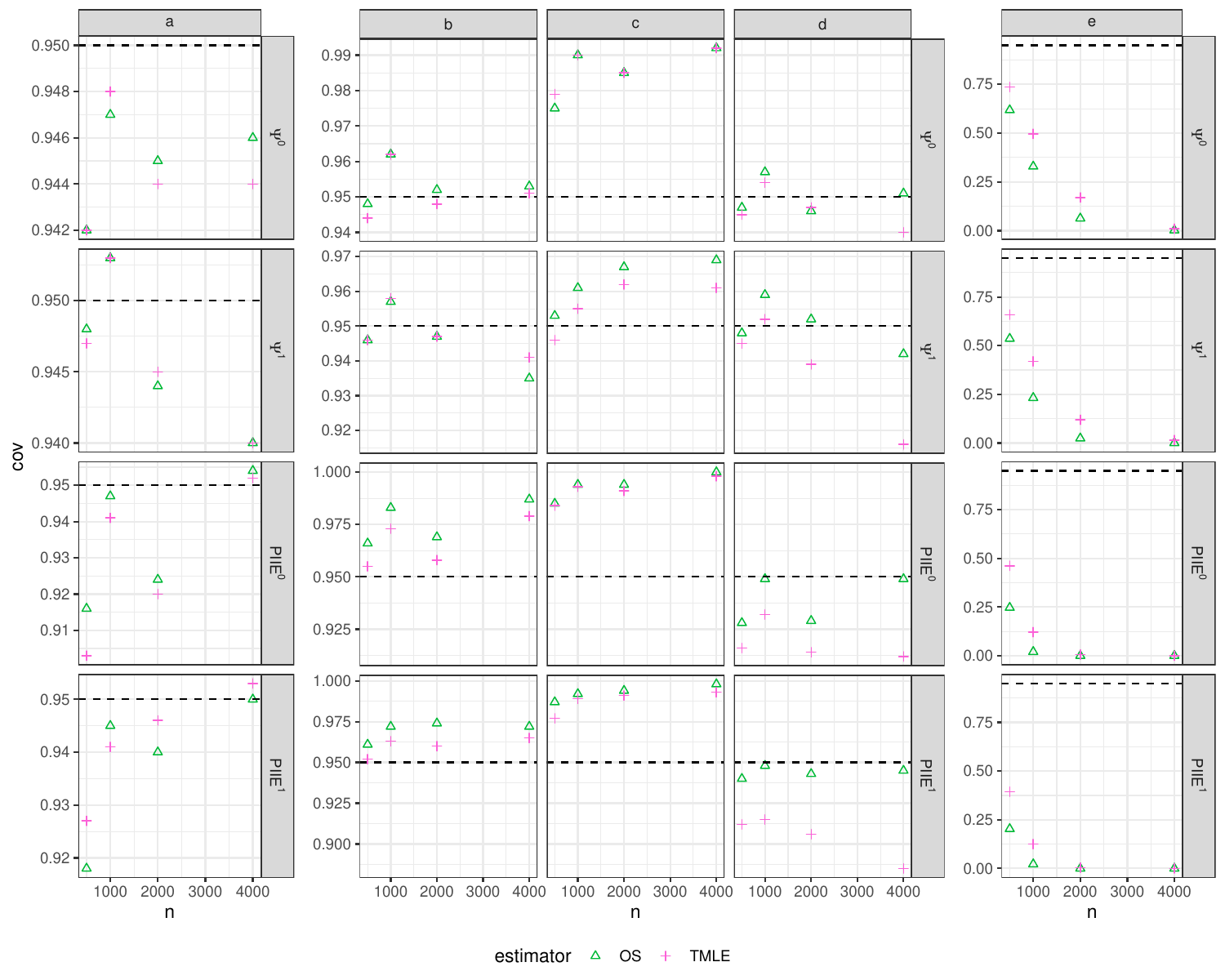}
    \caption{Simulation results for DGM (2). Coverage probability.}
    \label{fig:sim_cov2}
\end{figure}

\begin{figure}
    \centering
    \includegraphics[width=0.9\textwidth]{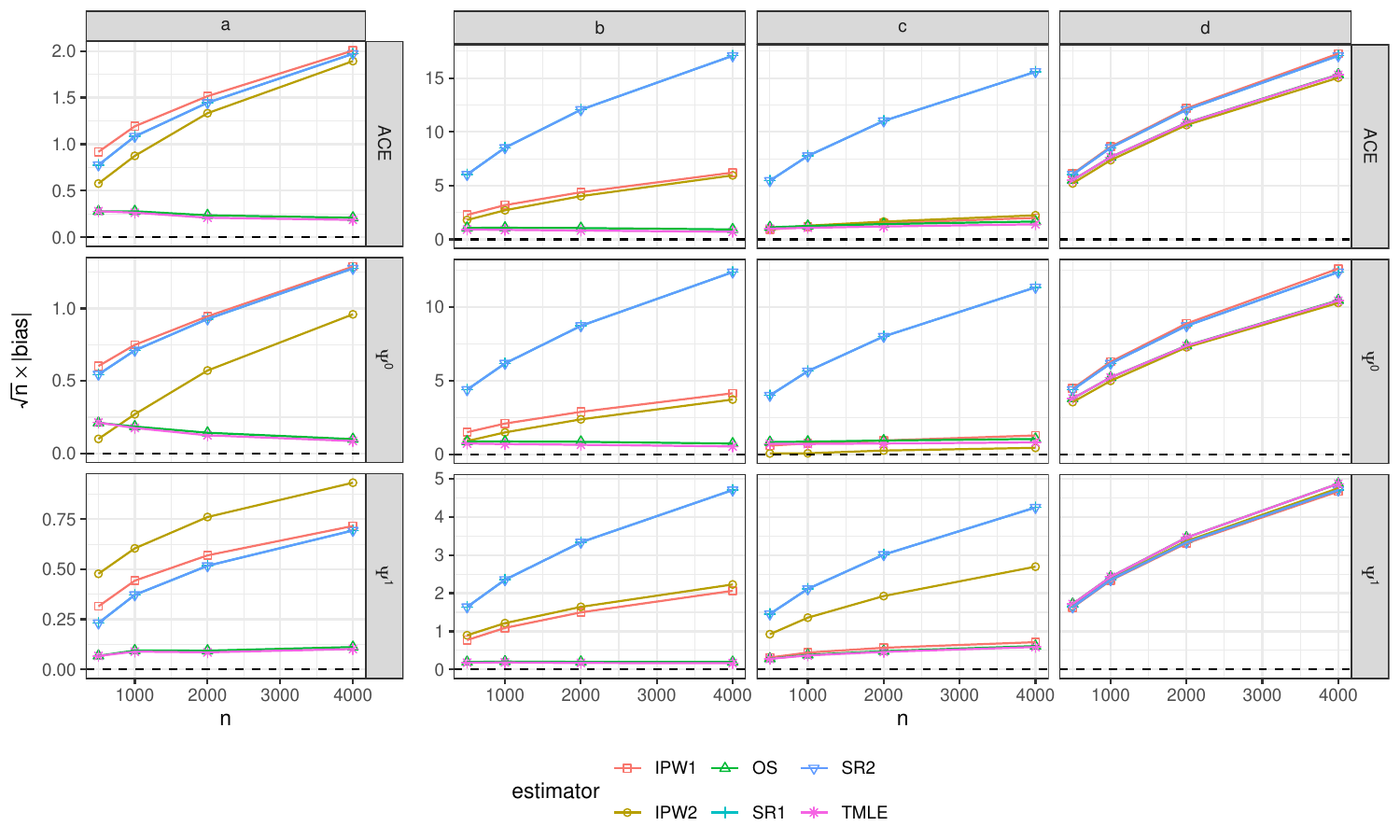}
    \caption{Simulation results for DGM (3). Absolute bias scaled by $\sqrt{n}$.}
    \label{fig:sim_bias3}
\end{figure}
\begin{figure}[H]
    \centering
    \includegraphics[width=0.9\textwidth]{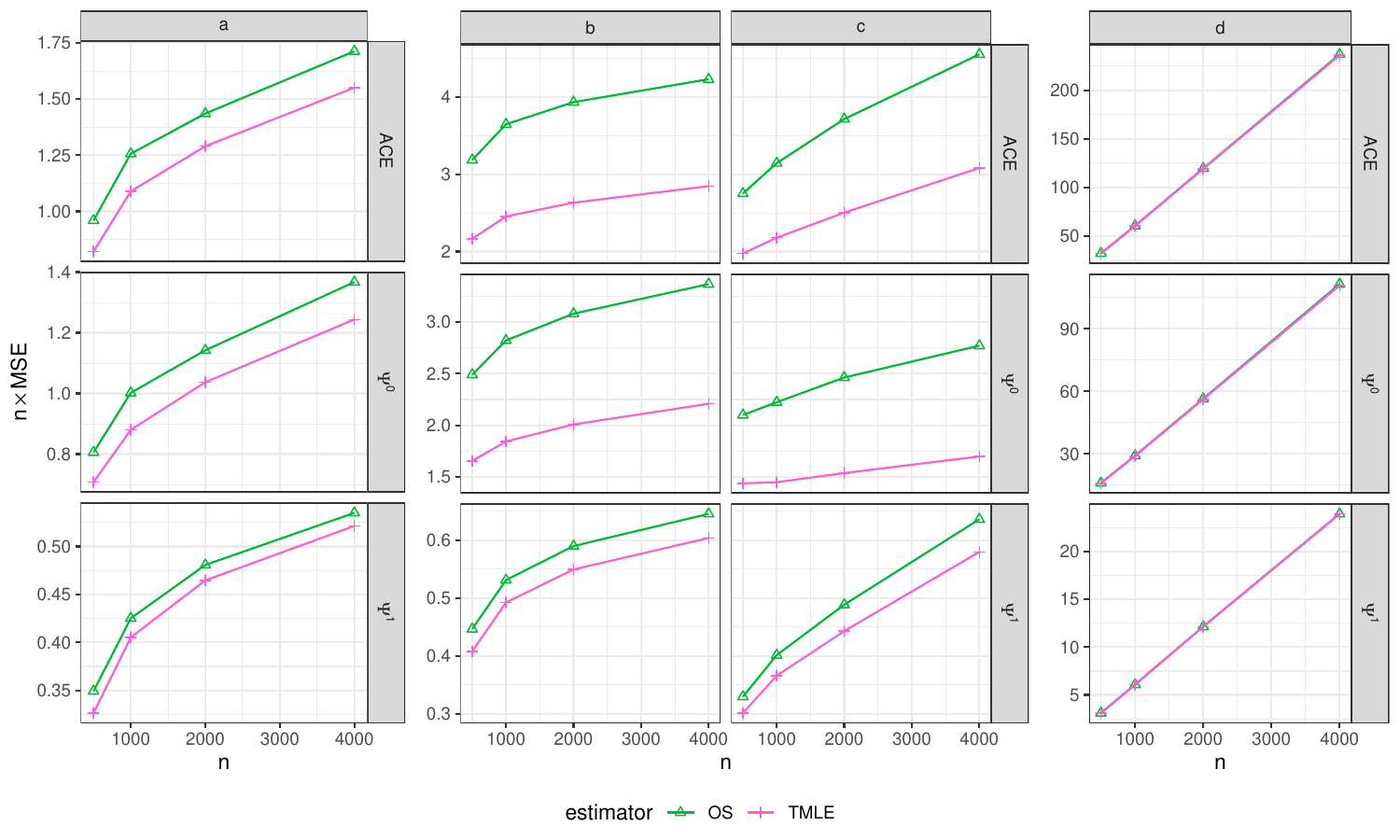}
    \caption{Simulation results for DGM (3). $\mbox{MSE}$ scaled by $n$.}
    \label{fig:sim_mse3}
\end{figure}

\begin{figure}[H]
    \centering
    \includegraphics[width=0.9\textwidth]{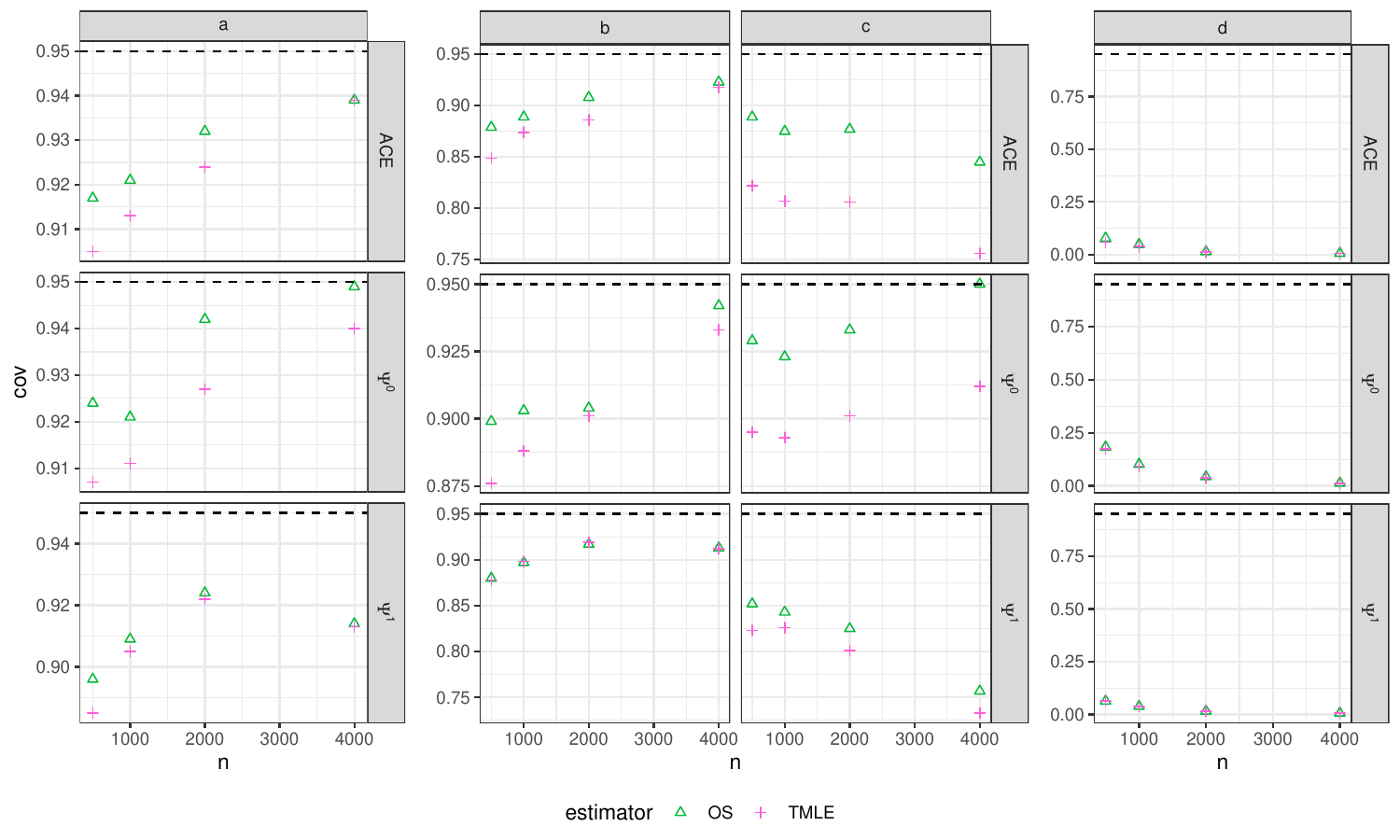}
    \caption{Simulation results for DGM (3). Coverage probability.}
    \label{fig:sim_cov3}
\end{figure}

\begin{figure}
    \centering
    \includegraphics[width=0.9\textwidth]{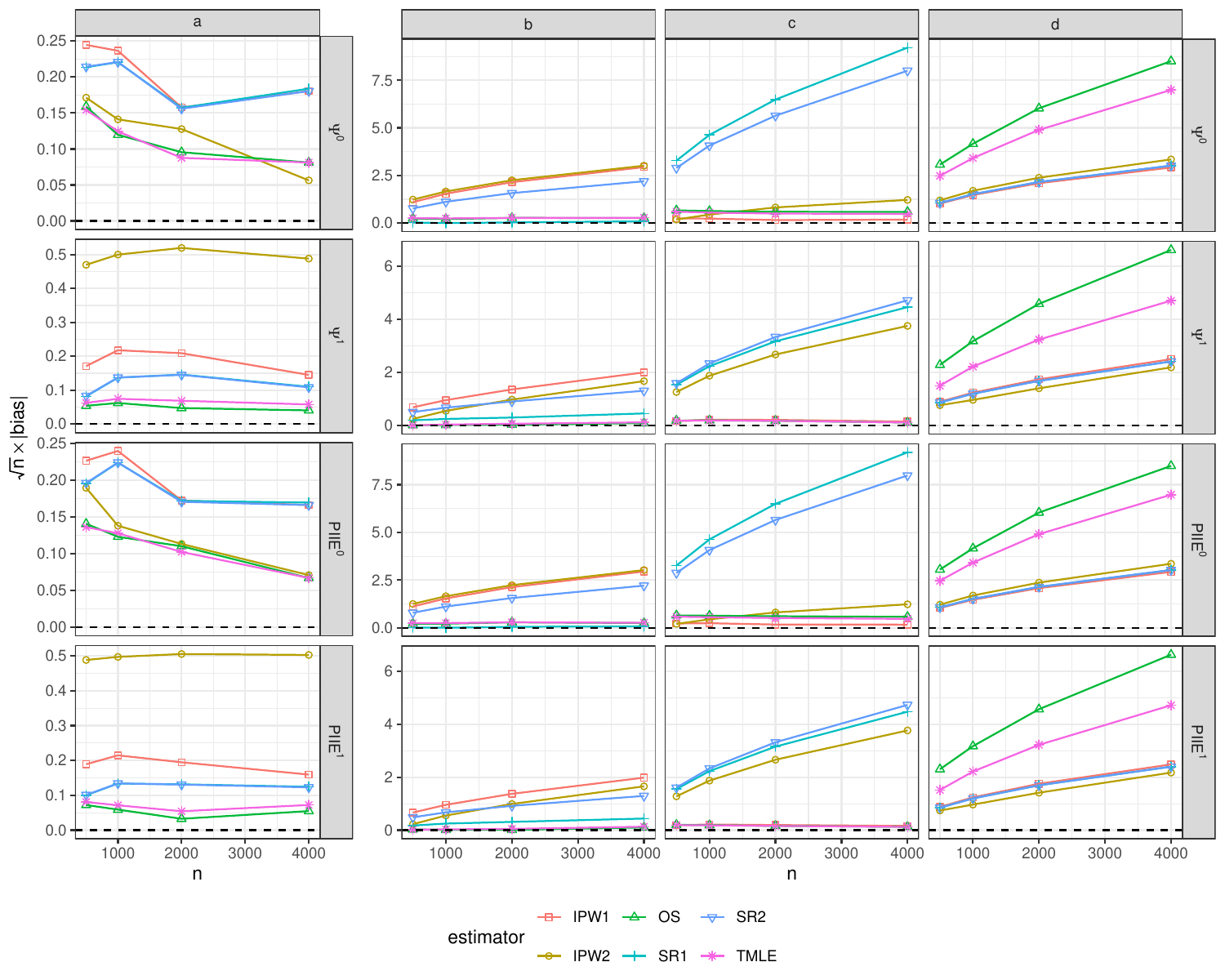}
    \caption{Simulation results for DGM (4). Absolute bias scaled by $\sqrt{n}$.}
    \label{fig:sim_bias4}
\end{figure}
\begin{figure}[H]
    \centering
    \includegraphics[width=0.9\textwidth]{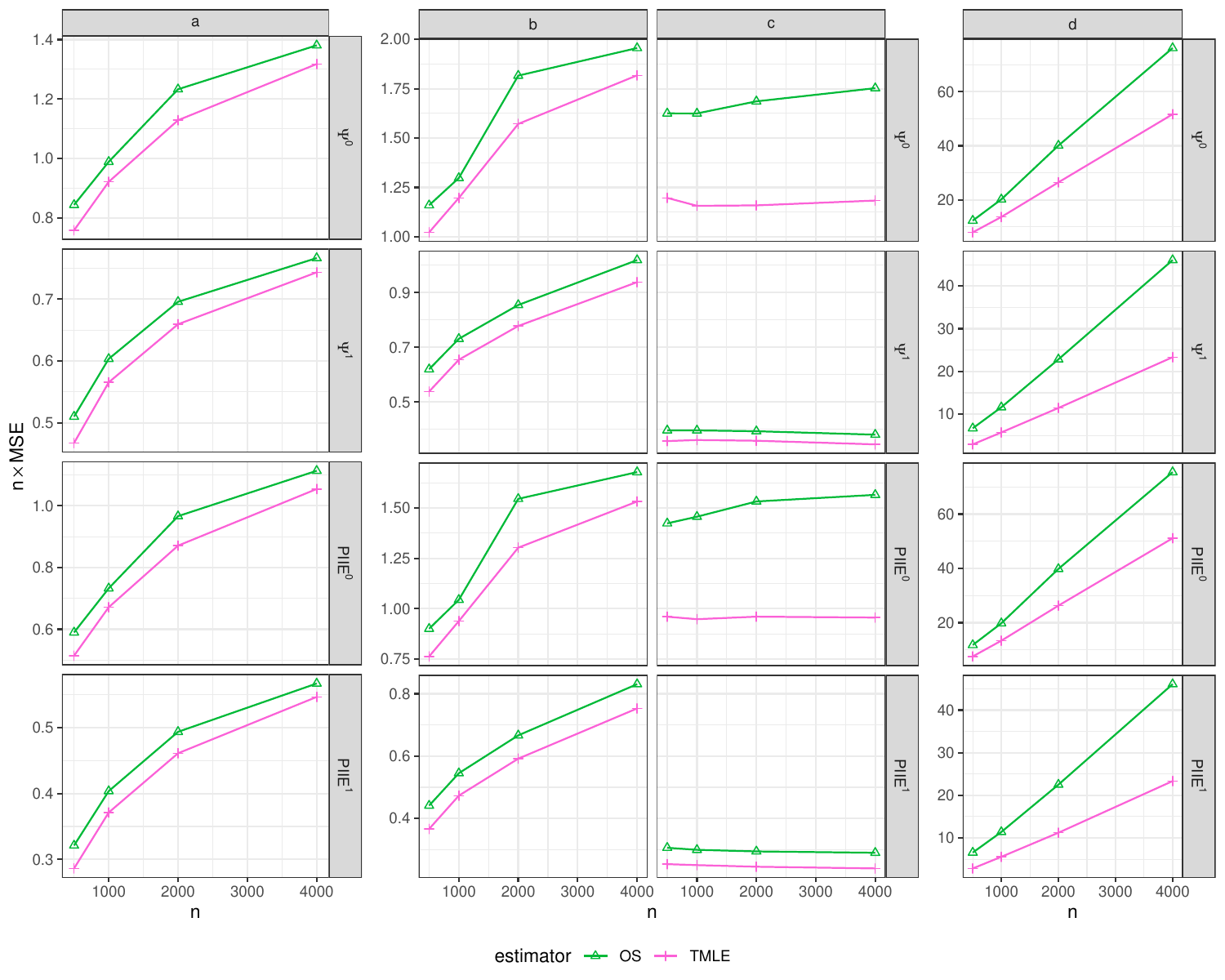}
    \caption{Simulation results for DGM (4). $\mbox{MSE}$ scaled by $n$.}
    \label{fig:sim_mse4}
\end{figure}

\begin{figure}[H]
    \centering
    \includegraphics[width=0.9\textwidth]{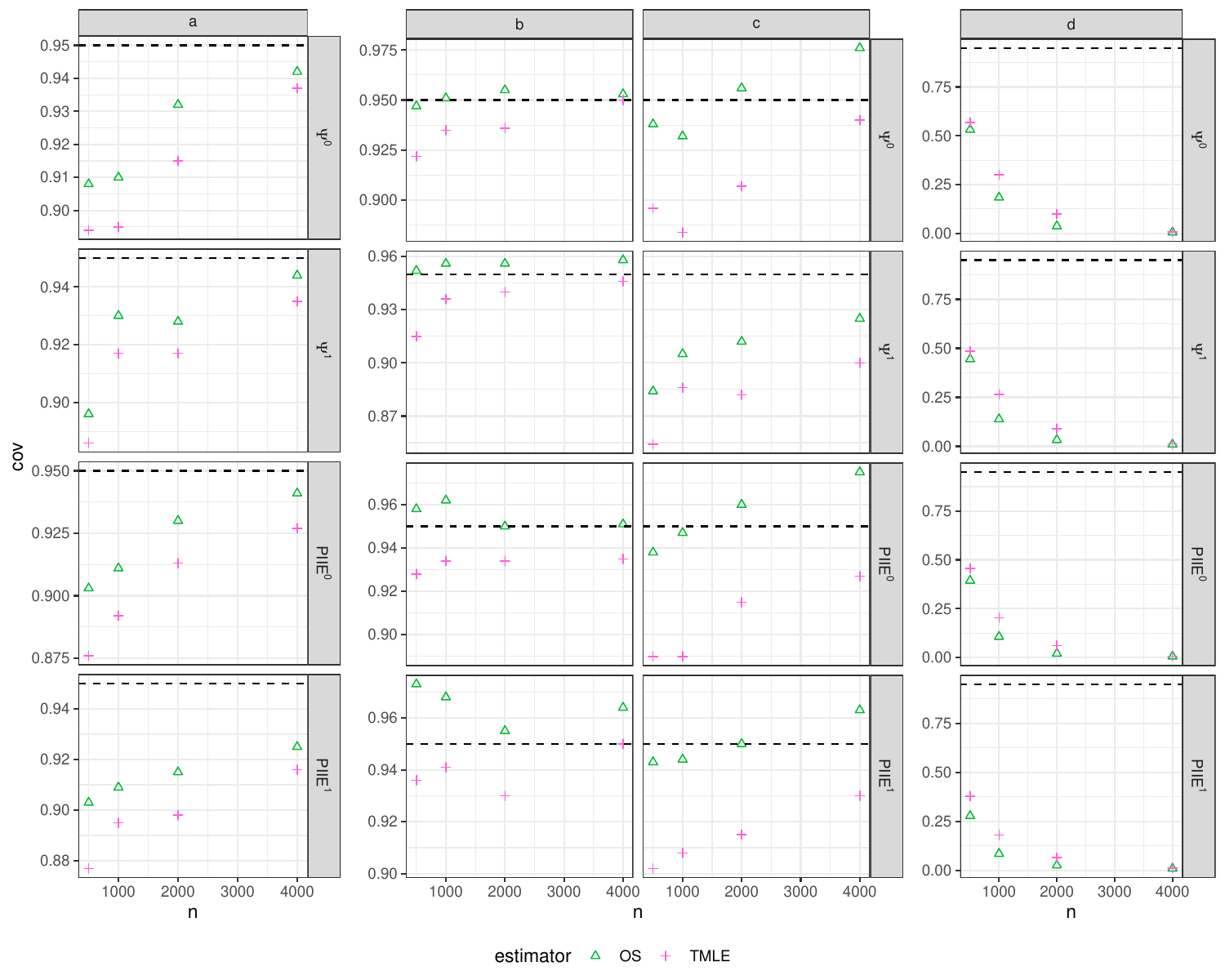}
    \caption{Simulation results for DGM (4). Coverage probability.}
    \label{fig:sim_cov4}
\end{figure}

\subsection{Super learner libraries}
The super learners were implemented using the \texttt{SuperLearner} R-package \citep{SuperLearner}. The library of candidate learners for the super learners was specified as follows:

\begin{lstlisting}
SL.library.screen <- list(
  "SL.mean",
  "SL.glm",
  "SL.glm.interaction",
  "SL.stepAIC",
  "SL.bayesglm",
  c("SL.glm", "screen.corP"),
  c("SL.glm.interaction", "screen.corP"),
  c("SL.step", "screen.corP"),
  c("SL.step.forward", "screen.corP"),
  c("SL.stepAIC", "screen.corP"),
  c("SL.step.interaction", "screen.corP"),
  c("SL.bayesglm", "screen.corP")
)
\end{lstlisting}


\end{document}